\documentclass[reqno]{amsart}
\usepackage{amsmath}
\usepackage{amsfonts}

\usepackage{amsmath,amssymb}
\usepackage{graphicx}
\usepackage{color}

\newtheorem{theorem}{Theorem}
\newtheorem{lemma}{Lemma}
\newtheorem{corollary}{Corollary}

\newtheorem*{conclusion*}{Conclusion}
\newtheorem*{definition*}{Definition}
\newtheorem*{remark*}{Remark}

\newtheorem{example}{Example}

\makeatletter

\newcommand{\Rmnum}[1]{\expandafter\@slowromancap\romannumeral #1@}
\makeatother

\newcommand{\notequiv}{{\,\not\equiv\, }}

\newcommand{\ls}[1]
    {\dimen0=\fontdimen6\the\font\lineskip=#1\dimen0
     \advance\lineskip.5\fontdimen5\the\font
     \advance\lineskip-\dimen0
     \lineskiplimit=0.9\lineskip
     \baselineskip=\lineskip
     \advance\baselineskip\dimen0
     \normallineskip\lineskip\normallineskiplimit\lineskiplimit
     \normalbaselineskip\baselineskip
     \ignorespaces}

\begin{document}

\bibliographystyle{abbrv}

\title{Cyclic codes from  the first class two-prime Whiteman's generalized cyclotomic sequence with order 6}
\author{Pramod Kumar Kewat and Priti Kumari}
\address{Department of Applied Mathematics, Indian School of Mines, Dhanbad 826 004, India}
\email{kewat.pk.am@ismdhanbad.ac.in, priti.jsr13@gmail.com}        
\date{}
 \maketitle
\markboth{P.K. Kewat and Priti Kumari}{Cyclic codes from  the first class two-prime WGCS with order 6}

\thispagestyle{plain} \setcounter{page}{1}

\begin{abstract}
Binary Whiteman's cyclotomic sequences of orders 2 and 4 have a number of good randomness properties. In this paper, we compute the autocorrelation values and linear complexity of the first class two-prime Whiteman's generalized cyclotomic sequence (WGCS-I) of order $d=6$. Our results show that the autocorrelation values of this sequence is four-valued or five-valued if $(n_1-1)(n_2-1)/36$ is even or odd respectively, where $n_1$ and $n_2$ are two distinct odd primes and their linear complexity is quite good. We employ the two-prime WGCS-I of order 6 to construct several classes of cyclic codes over $\mathrm{GF}(q)$ with length $n_1n_2$. We also obtain the lower bounds on the minimum distance of these cyclic codes.
\end{abstract}

\section{Introduction}\label{section 1}
Let $q$ be a power of a prime $p$. An $[n,k,d]$ linear code $C$ over a 
finite field $\mathrm{GF}(q)$ is a $k-$dimensional subspace of the vector space $\mathrm{GF}(q)^n$ with the minimum distance $d$.
A linear code $C$ is a cyclic code if the cyclic shift of a codeword in $C$ is again a codeword in $C$, i.e., 
if $ (c_0,\cdots,c_{n-1})\in C$ then $(c_{n-1},c_0\cdots,c_{n-2})\in C$.
Let $\mathrm{gcd}(n, q)= 1$. We consider the univariate polynomial ring $\mathrm{GF}(q)[x]$ and the ideal 
$I=\langle x^n-1\rangle$ of $\mathrm{GF}(q)[x].$ We denote by $R$ the ring  $\mathrm{GF}(q)[x]/I$.
We can consider a cyclic code of length $n$ over $\mathrm{GF}(q)$ as an ideal in $R$ via the following correspondence
\begin{align}
 \mathrm{GF}(q)^n \rightarrow  R,\ \ \ \ \ (c_0,c_1,\cdots,c_{n-1}) \mapsto c_0+c_1x+\cdots+c_{n-1}x^{n-1}.\nonumber
\end{align}
 Then, a linear code $C$ over $\mathrm{GF}(q)$ is a cyclic code if and only if $C$ is an ideal of $R$. Since $R$ is a principal ideal ring, if $C$ is not trivial, there exists a unique monic polynomial
 $g(x)$ dividing $x^n-1$ in $\mathrm{GF}(q)[x]$ and $C=\langle g(x) \rangle$. The polynomials $g(x)$ and $h(x)=(x^n-1)/g(x)$ are called
 the generator polynomial and the parity-check polynomial of $C$ respectively.
 If the dimension of the code $C$ is $k$, the generator polynomial has degree $n-k.$  An $[n,k,d]$ cyclic code $C$ is capable of encoding $q-$ary messages of length $k$ and requires $n-k$ 
 redundancy symbols.
 
 The total number of cyclic codes over $\mathrm{GF}(q)$ and their construction are closely related to the cyclotomic cosets modulo $n$. 
 One way to construct cyclic codes over $\mathrm{GF}(q)$ with length $n$ is to use the generator polynomial
\begin{align}
 \frac{x^n-1}{\mathrm{gcd}(x^n-1,S(x))}\label{genpol},
\end{align}
where $S(x)=\sum\limits_{i=0}^{n-1}s_ix^i \in \mathrm{GF}(q)[x]$ and $s^\infty={(s_i)}_{i=0}^{\infty}$ is a sequence of period $n$ over $\mathrm{GF}(q).$
The cyclic code $C_s$ generated by the polynomial in Eq.(\ref{genpol}) is called the cyclic code defined by the sequence $s^{\infty},$ and the sequence $s^{\infty}$ is called the defining sequence of the cyclic code $C_s.$ 

Cyclic codes have been studied in a series of papers and a lot of progress have been accomplished 
 (see, for example \cite{Betti06}, \cite{Ding15}, \cite{Eupen93},  \cite{Macwilliams77} and \cite{Lint86}).
 The Whiteman's generalized cyclotomy was introduced by Whiteman and its properties were studied in \cite{whiteman62}.
 The two-prime Whiteman's generalized cyclotomic sequence(WGCS) was introduced by Ding \cite{ding97} and its coding properties were studied in \cite{ding12} and \cite{white04}.
 Ding \cite{ding12} and Sun et al.\cite{white04} constructed number of classes of
cyclic codes over $\mathrm{GF}(q)$ with length product of two-distict primes from the two-prime WGCS of order $2$ and $4$ respectively and gave the lower bounds
on the minimum weight of these cyclic codes under certain conditions. After that Kewat et.al \cite{kewat} constructed number of classes of cyclic codes over $\mathrm{GF}(q)$ with length $n$ from the second class two-prime Whiteman's generalized cyclotomic sequence (WGCS-II) of order 6. In this paper, we employ the first class two-prime Whiteman's generalized cyclotomic sequence (WGCS-I) of order 6 to construct number of classes of cyclic codes over $\mathrm{GF}(q)$. 

The autocorrelation values of generalized cyclotomic sequences of order 2, 4 and 6 have been studied in (\cite{auto2},\cite{auto4}, and \cite{Auto13}). It has been shown that the autocorrelation values of the generalized cyclotomic sequence of order two are quite good when $|q-p|$ is small enough. In case of order 4 and 6, the results show that the autocorrelation values are low provided that the parameters are chosen carefully. First, we discuss about the autocorrelation properties of the two-prime WGCS-I.
Binary sequences with low correlation values have wide applications in stream cipher, software testing, communication systems, radar navigation and other fields. In this correspondence, we calculate the exact autocorrelation values of the two-prime WGCS-I. The linear complexity of generalized cyclotomic sequences of length product of two-distict primes has been calculated by Ding \cite{ding97} and Bai \cite{bai05}.
We also calculate the exact value of the linear complexity of this sequence without any special requirements on the primes. We employ the two-prime WGCS-I with order $6$ to construct several classes of cyclic codes over $\mathrm{GF}(q).$  We also obtain the lower bounds 
on the minimum weight of these cyclic codes.
 
  \section{Preliminaries}\label{section 2}
\subsection{Linear complexity and minimal polynomial}\label{subsection 2.1}

If ${(s_i)}_{i=0}^{\infty}$ is a sequence over a finite field $\mathrm{GF}(q)$ and $f(x)$ is a polynomial with coefficients in $\mathrm{GF}(q)$ given by
$ f(x)=c_0+c_1x+\cdots+c_{L-1}x^{L-1},$
then we define 
$ f(E)s_j=c_0s_j+c_1s_{j-1}+\cdots+c_{L-1}s_{j-L+1},$
where $E$ is a left shift operator defined by $Es_i=s_{i-1}$ for $i\geq 1.$
Let $s^n$ be a sequence $s_0s_1\cdots s_{n-1}$ of length $n$ over a finite field $\mathrm{GF}(q)$. For a finite sequence, the $n$ is finite; for a semi-infinite sequence, the $n$ is $\infty$.
A polynomial $f(x)\in \mathrm{GF}(q)[x]$ of degree $\leqslant l$ with $c_0\neq 0$ is called a characteristic polynomial of the sequence $s^n$ if $f(E)s_j=0$ for all $j$ with $j\geq l.$
For every characteristic polynomial there is a least $l\geq \mathrm{deg}(f)$ such that the above equation hold. The smallest $l$  is called the associate recurrence length of $f(x)$
with respect to the sequence $s^n$. The characteristic polynomial with smallest length is known as minimal polynomial of the sequence $s^n$ and the associated recurrence 
length is called the linear span or linear complexity of the sequence $s^n$.\\ If a semi-infinite sequence $s^{\infty}$ is periodic, then its minimal polynomial is unique 
if $c_0=1.$ The linear complexity of a periodic sequence is equal to the degree of its minimal polynomial. For the periodic sequences,
there are few ways to determine their linear spans and minimal polynomials. One of them is given in the following lemma.
\begin{lemma}\label{spanmin}
\cite{finite} Let $s^\infty$ be a sequence of a period $n$ over $GF(q)$. Define
 \begin{align}
 S^{n}(x)=s_0+s_{1}x +\cdots+s_{n-1}x^{n-1}\in \mathrm{GF} (q) [x].\nonumber
 \end{align}
 Then the minimal polynomial $m_s$ of $s^\infty$ is given by
 \begin{align}
  \frac{x^n-1}{\mathrm{gcd}(x^n-1,S^n(x))}.\nonumber 
 \end{align}
 Consequently, the linear span $L_s$ of $s^\infty$ is given by
 \begin{align}
  L_s=n-deg(\mathrm{gcd}(x^n-1,S^n(x))).\nonumber 
 \end{align}
\end{lemma}

\subsection{The Whiteman's generalized cyclotomic sequences and its construction}\label{subsection 2.2}
Let $n$ be a positive integer. An integer $a$ is called a primitive root of modulo $n$ if the multiplicative order of $a$ modulo $n,$ denoted by $ \mathrm{ord}_n(a)$, is equal to $\phi(n),$ where $\phi(n)$ is the Euler phi function and $\mathrm{gcd}(a,n)=1.$ 
Let $n_1$ and $n_2$ be two distinct odd primes, define $n=n_1n_2,\ d=\mathrm{gcd}(n_1-1,n_2-1)$ and $e=(n_1-1)(n_2-1)/d$. From the Chinese Remainder theorem, there are common primitive roots of both $n_1$ and $n_2$. Let $g$ be a fixed common primitive root of both $n_1$ and $n_2$. Let $u$ be an integer satisfying
\begin{align}
 u\equiv g \ (\mathrm{mod}\ n_1),~~~~~  u\equiv 1 \ (\mathrm{mod}\ n_2).\label{defu}
\end{align}
Whiteman \cite{whiteman62} proved that
\begin{align}
 \mathbb{Z}_{n}^{\ast}=\{g^{s}u^{i}:s=0,1,\cdots,e-1;\ i=0,1,2,\cdots,d-1\}.\nonumber
\end{align}
where $\mathbb{Z}_{n}^{\ast}$ denotes the set of all invertible elements of the residue class ring $\mathbb{Z}_{n}$ and $e$ is the order of $g$ modulo $n$.
The Whiteman's generalized cyclotomic classes $W_i$  of order $d$ are defined by
\begin{align}
 W_i=\{g^su^i ~(\mathrm{mod}\ n) :s=0,1,\cdots,e-1\},i=0,1,\cdots,d-1.\nonumber
\end{align}
The classes $W_i$, $1\leq i\leq d-1$ give a partition of $ \mathbb{Z}_{n}^{\ast},$ i.e.,
 $\mathbb{Z}_{n}^{\ast}=\cup_{i=0}^{d-1} {W_i},\ W_{i}\cap W_{j}=\emptyset \ \mathrm{for}\ i\neq j.$ \\
Let
$$P=\{n_{1},2n_{1},3n_{1},\cdots,(n_{2}-1)n_{1}\},\ Q=\{n_{2},2n_{2},3n_{2},\cdots,(n_{1}-1)n_{2}\},$$
$$C_{0}=\{0\}\cup Q\cup \bigcup\limits_{i=0}^{\frac{d}{2}-1}W_{2i}\ and \ \ C_{1}=P\cup \bigcup\limits_{i=0}^{\frac{d}{2}-1}W_{2i+1},$$
$$C_{0}^{\ast}=\{0\}\cup Q\cup \bigcup\limits_{i=0}^{\frac{d}{2}-1}W_{i},\ C_{1}^{\ast}=P\cup \bigcup\limits_{i=\frac{d}{2}}^{d-1}W_{i}.$$

It is easy to see that if $d > 2$, then $C_0 \neq C_0^{\ast}$ and $C_1 \neq C_1^{\ast}$ . 
Now we introduce two kinds of Whiteman's generalized cyclotomic sequences of order $d$ (see \cite{defwhiteman}).
\begin{definition*} The two-prime Whiteman's generalized cyclotomic sequence $\lambda^\infty= (\lambda_i)_{i=0}^{n-1} $ of order $d$ and period $n$, 
which is called the first class two-prime Whiteman's generalized cyclotomic sequence denoted by WGCS-I, is defined by
\begin{equation}
\lambda_{i}=\left\{
\begin{array}{ll}
0,\ \ \mathrm{if}\ i\in C_{0}\\
1,\ \ \mathrm{if}\ i\in C_{1}.\label{si}
\end{array}
\right.
\end{equation}
The two-prime Whiteman's generalized cyclotomic sequence $s^\infty=(s_i)_{i=0}^{n-1}$ of order $d$ and period $n$, 
which is called the second class two-prime Whiteman's generalized cyclotomic sequence denoted by  WGCS-II, is defined by
\begin{equation}
s_{i}=\left\{
\begin{array}{ll}
0,\ \ \mathrm{if}\ i\in C^{\ast}_{0}\\
1,\ \ \mathrm{if}\ i\in C^{\ast}_{1}.\nonumber
\end{array}
\right.
\end{equation}
\end{definition*}
The sets $C_1$ and $C_1^\ast\subseteq \mathbb{Z}_n$ are called the characteristic sets of the sequences $\lambda^\infty$ and $s^\infty$ respectively, and the sequences $\lambda_i$ and $s_i$ are referred as the characteristic sequences of $C_1$ and $C_1^\ast$ respectively.

The cyclotomic numbers corresponding to these cyclotomic classes are defined as
$$
(i,j)_{d}=|(W_{i}+1)\cap W_{j}|, \nonumber where~~ 0\leq i,j\leq d-1.
$$
Additionally, for any $t\in \mathbb{Z}_{n}$, we define
\begin{align}
d(i,j;t)=|(W_{i}+t)\cap W_{j}|,\nonumber
\end{align}
where $W_{i}+t=\{w+t|w\in W_{i}\}$.

\subsection{Properties of Whiteman's cyclotomy of order d}\label{section 3.1}
In this subsection, we summarize number of properties of Whiteman's generalized cyclotomy of order $d=\mathrm{gcd}(n_1-1,n_2-1).$
The proof of the following Lemma follows from the Theorem 4.4.6 of \cite{Streamcipher}.
\begin{lemma}\label{val d}
Let the notations be same as before and $t\neq0$. We have
\begin{align}
d(i,j;t)=\left\{
\begin{array}{llll}
\frac{(n_{1}-1)(n_{2}-1)}{d^2},\ \ \ \ \ \ \ \ \ \ \ \ \ \ \ \ \  i\neq j,\ t \in  P \cup Q  \nonumber\\
\frac{(n_{1}-1)(n_{2}-1-d)}{d^2},\ \ \ \ \ \ \ \ \ \ \ \ \ \ i=j,\ t\in P,\ t \notin Q \nonumber \\
\frac{(n_{1}-1-d)(n_{2}-1)}{d^2},\ \ \ \ \ \ \ \ \ \ \ \ \ \ i=j,\ t\in Q,\ t \notin P \nonumber\\
(i',j')_d \ for\  some\ (i',j'),\ \  otherwise.
\end{array}
\right.
\end{align}
\end{lemma}

\begin{lemma}\label{W3}
Let the symbols be defined as before. The following four statements are equivalent:\\
(1) $-1 \in W_{\frac{d}{2}}.$\\
(2)  $\frac{(n_1-1)(n_2-1)}{d^2}$ is even.\\
(3) One of the following sets of equations are satisfied:\\
 $
 \left\{
\begin{array}{ll}
 n_1 \equiv 1 \ (\mathrm{mod}\ 2d)\nonumber\\
 n_2 \equiv d+1 \ (\mathrm{mod}\ 2d),\nonumber
\end{array}
\right.
 \left\{
\begin{array}{ll}
 n_1 \equiv d+1 \ (\mathrm{mod}\ 2d)\nonumber\\
 n_2 \equiv 1 \ (\mathrm{mod}\ 2d).\nonumber
\end{array}
\right.
$\\
(4) $ n_1n_2 \equiv d+1 \ (\mathrm{mod}\ 2d).$\\
\end{lemma}

\begin{proof}
$(1) \Leftrightarrow (2)$  The result follows from (2.3) in \cite{whiteman62}.\\
$(2) \Rightarrow (3) $  Let $n_1-1=df,\  n_2-1=df'$ and $e=dff'$, where $f$ and $f'$ are integer. Since $\mathrm{gcd}(f,f')=1,\ f$ and $f'$ can not both be even. Here $ff'=\frac{(n_1-1)(n_2-1)}{d^2}$ is even. So, $f$ or $f'$ is even.
Let $f$ is even and $f'$ is odd. If $f$ is even, then $n_1-1 = d(2k_1)$, where $k_1$ is an integer. Therefore, $n_1 \equiv 1 \ ( \mathrm{mod}\ 2d )$.  If $f'$  is odd, then $n_1-1 = d(2k_2+1)$, where $k_2$ is an integer. Therefore, $n_2 \equiv d+1 \ ( \mathrm{mod}\ 2d ).$ 
Similarly, when $f$ is odd and $f'$ is even. We get $n_1 \equiv d+1 \ ( \mathrm{mod}\ 2d )$ and $n_2 \equiv 1 \ ( \mathrm{mod}\ 2d ).$\\
$(3) \Rightarrow (2)$  and $(3) \Rightarrow (4)$ are obvious.\\
$(4) \Rightarrow (3)$  Since $\mathrm{gcd}(n_1-1,n_2-1)=d$, let $n_1-1=fd$ and $n_2-1=f'd$. We have $ n_1n_2 \equiv d+1\ (\mathrm{mod}\ 2d)$, this gives $fd+f'd\equiv d \ (\mathrm{mod}\ 2d).$
Thus, we have $f+f'=2k+1$ for an integer $k$. So, $n_1=2kd+(1-f')d+1$, this gives $n_1\equiv(1-f')d+1 \ (\mathrm{mod}\ 2d).$
If $f'$ is odd, then $n_1\equiv 1 \ (\mathrm{mod}\ 2d)$ and $n_2 \equiv d+1 \ (\mathrm{mod}\ 2d)$. If $f'$ is even, then $n_1\equiv d+1 \ (\mathrm{mod}\ 2d)$ and $n_2 \equiv 1 (\mathrm{mod}\ 2d).$ 
\end{proof}

\begin{lemma}\label{W0}
Let the symbols be defined as before. The following four statements are equivalent:\\
(1) $-1 \in W_0.$\\
(2)  $\frac{(n_1-1)(n_2-1)}{d^2}$ is odd.\\
(3)  The following set of equation is satisfied:\\
$
\left\{
\begin{array}{ll}
n_1 \equiv d+1 \ (\mathrm{mod}\ 2d)\nonumber\\
n_2 \equiv d+1 \ (\mathrm{mod}\ 2d),\nonumber
\end{array}
\right.
$\\
(4) $ n_1n_2 \equiv (d+1)^2=1 \ (\mathrm{mod}\ 2d).$
\end{lemma}

\begin{proof}
The proof is similar to the proof of the above Lemma.
\end{proof}

\subsection{Properties of Whiteman's cyclotomy of order 6}\label{section 3.2}
We recall the following lemmas (Lemma 1 and Lemma 2) from \cite{gong13}. 
\begin{lemma}\label{odd&even}
Let $\mathrm{gcd}(n_1-1,n_2-1) = 6$, i.e.,  $n_1 \equiv 1\ \mathrm{mod\ 6}$, $n_2 \equiv 1\ \mathrm{mod\ 6}$. 
Let $a,b,x,y,c$ and $d$ are integers. There are 10 possible different cyclotomic numbers of order 6 and they are given by the following relations:
\\ If $\frac{(n_{1}-1)(n_{2}-1)}{36}$ is odd, we have
 $$
 \begin{array}{lllll}
 &(0,0)_{6}=\frac{1}{72}(12M+32+6a-24x+2c),\\
 &(0,1)_{6}=(1,0)_{6}=(5,5)_{6}=\frac{1}{72}(12M+8+a+3b+8x+24y-c+9d),\\
 &(0,2)_{6}=(2,0)_{6}=(4,4)_{6}=\frac{1}{72}(12M+8-3a+9b-c-9d),\\
 &(0,3)_{6}=(3,0)_{6}=(3,3)_{6}=\frac{1}{72}(12M+8-2a+8x+2c),\\
 &(0,4)_{6}=(4,0)_{6}=(2,2)_{6}=\frac{1}{72}(12M+8-3a-c-9b+9d),\\
 &(0,5)_{6}=(5,0)_{6}=(1,1)_{6}=\frac{1}{72}(12M+8+a-3b+8x-24y-c-9d),\\
 &(1,2)_{6}=(2,1)_{6}=(4,5)_{6}=(5,4)_{6}=(5,1)_{6}=(1,5)_{6}=\frac{1}{72}(12M-4-2a-4x+2c),\\
 &(1,3)_{6}=(2,5)_{6}=(3,1)_{6}=(3,4)_{6}=(4,3)_{6}=(5,2)_{6}=\frac{1}{72}(12M-4+a+3b-4x-12y-c+9d),\\
 &(1,4)_{6}=(2,3)_{6}=(3,2)_{6}=(3,5)_{6}=(4,1)_{6}=(5,3)_{6}=\frac{1}{72}(12M-4+a-3b-4x+12y-c-9d) \ \mathrm{and} \\
 &(2,4)_{6}=(4,2)_{6}=\frac{1}{72}(12M-4+6a+12x+2c).
 \end{array}
 $$
 If $\frac{(n_{1}-1)(n_{2}-1)}{36}$ is even, we have
 $$
 \begin{array}{lllll}
 &(0,0)_{6}=(3,0)_{6}=(3,3)_{6}=\frac{1}{72}(12M+20-8x-2a+2c),\\
 &(0,1)_{6}=(2,5)_{6}=(4,3)_{6}=\frac{1}{72}(12M-4-3a-9b-c+9d),\\
 &(0,2)_{6}=(1,4)_{6}=(5,3)_{6}=\frac{1}{72}(12M-4-8x+a-c+24y-3b-9d),\\
 &(0,3)_{6}=\frac{1}{72}(12M-4+24x+6a+2c),\\
 &(0,4)_{6}=(1,3)_{6}=(5,2)_{6}=\frac{1}{72}(12M-4-8x+a-c-24y+3b+9d),\\
 &(0,5)_{6}=(2,3)_{6}=(4,1)_{6}=\frac{1}{72}(12M-4-3a-c+9b-9d),\\
 &(1,0)_{6}=(2,2)_{6}=(3,1)_{6}=(3,4)_{6}=(4,0)_{6}=(5,5)_{6}=\frac{1}{72}(12M+8+4x+a-c+12y+3b+9d),\\
 &(1,1)_{6}=(2,0)_{6}=(3,2)_{6}=(3,5)_{6}=(4,4)_{6}=(5,0)_{6}=\frac{1}{72}(12M+8+4x+a-c-12y-3b-9d),\\
 &(1,2)_{6}=(1,5)_{6}=(2,4)_{6}=(4,2)_{6}=(5,1)_{6}=(5,4)_{6}=\frac{1}{72}(12M-4+4x-2a+2c)\ \mathrm{and} \\
 &(2,1)_{6}=(4,5)_{6}=\frac{1}{72}(12M-4+6a-12x+2c).
 \end{array}
 $$
 Where $n_1n_2 = x^2+3y^2, M = \frac{1}{6}((n_1-2)(n_2-2)-1)$ and $4n_1n_2 = a^2+3b^2 = c^2+27d^2$. 
 \end{lemma}
 
 \begin{lemma}\label{lemma 3.6}
\cite{whiteman62} Define $\eta=\frac{(n_{1}-1)(n_{2}-1)}{36}$. Let symbols be same as before. Then
$$
-1 \in \left\{
\begin{array}{llll}
W_0,\   \ \mathrm{if}\ \eta\   is\ odd\\
W_3,\   \ \mathrm{if}\ \eta\   is\ even.
\end{array}
\right.
$$
\end{lemma}

\section{Autocorrelation values}\label{aucorrelation}
In this section, we will calculate the autocorrelation values of two-prime WGCS-I.
Let the symbols be the same as before. The periodic autocorrelation function of the binary sequence $\lambda^\infty$ of period $n$ is defined by
\begin{align}
 C_\lambda(w)=\frac{1}{n}\sum\limits_{i=0}^{n-1}(-1)^{\lambda_{i+w}+\lambda_i},\nonumber
\end{align}
where $0 \leqslant w \leqslant n-1.$
We define,\\
$d_\lambda(i,j;w)=|C_i\cap (C_j+w)|,$ \ \ $0 < w\leqslant n-1,\ i,j=0,1.$

\begin{lemma}\label{val d1}
 \cite{auto2} Let ${\lambda_i}$ be the characteristic sequence of $C_1 \subseteq \mathbb{Z}_n.$ Then we have
\begin{align}
 C_\lambda(w)=1-4\frac{d_\lambda(1,0;w)}{n},\nonumber
\end{align}
where $0 < w \leqslant n-1.$
\end{lemma}
The following formula will be needed in the sequel:
\begin{align}
&d_\lambda(1,0;w)=|C_1\cap (C_0+w)| \nonumber \\
&~~~~~~~~~~~~~~~~~=|(P\cup W_1 \cup W_3\cup W_5)\cap ((R\cup Q\cup W_0\cup W_2\cup W_4)+w)| \nonumber \\
&=|(W_1\cup W_3\cup W_5)\cap ((W_0\cup W_2\cup W_4)+w)|+|P\cap ((W_0\cup W_2\cup W_4)+w)| \nonumber \\
&+|(W_1\cup W_3\cup W_5)\cap ((R\cup Q)+w)|+|P\cap ((R \cup Q)+w)|.\label{ds}
\end{align}
To calculate the value of $d_\lambda(1,0;w)$, we need the following lemmas.

\begin{lemma}\label{val}
 For each $w\in {\mathbb{Z}^*_n}$, we have\\
\rm{(i)} If $(n_1-1)(n_2-1)/36$ is even, then
\begin{align}
|(W_1\cup W_3\cup W_5)\cap ((W_0\cup W_2\cup W_4)+w)|=\left\{ 
 \begin{array}{ll}
  (3M-1)/2,\ \ if\  w\in W_0\cup W_2\cup W_4\nonumber\\
  (3M+1)/2,\ \ if\  w\in W_1\cup W_3\cup W_5\nonumber.\\
 \end{array}
 \right.
\end{align}
\rm{(ii)} If $(n_1-1)(n_2-1)/36$ is odd, then
\begin{align}
|(W_1\cup W_3\cup W_5)\cap ((W_0\cup W_2\cup W_4)+w)|=\left\{ 
 \begin{array}{ll}
  3M/2,\ \ if\  w\in {\mathbb{Z}^*_n}. \nonumber\\
 \end{array}
 \right.
\end{align}

\end{lemma}
\begin{proof}
We know that 
$|(W_1\cup W_3\cup W_5)\cap ((W_0\cup W_2\cup W_4)+w)|=|W_1\cap (W_0+w)|+|W_3\cap (W_0+w)|+|W_5\cap (W_0+w)|+|W_1\cap (W_2+w)|+|W_3\cap (W_2+w)|+|W_5\cap (W_2+w)|+|W_1\cap (W_4+w)|+|W_3\cap (W_4+w)|+|W_5\cap (W_4+w)|$.\\
From Lemma 8 in \cite{Auto13}, we have the following results:\\
\rm{(a)}\ \ \ $
|W_1\cap (W_0+w)|=\left\{
\begin{array}{llll}
(0,1)_6,\ \ if\ \ w\in W_0\nonumber\\
(5,0)_6,\ \ if\ \ w\in W_1\nonumber \\
(4,5)_6,\ \ if\ \ w\in W_2\nonumber\\
(3,4)_6,\ \ if\ \ w\in W_3\nonumber\\
(2,3)_6,\ \ if\ \ w\in W_4\nonumber\\
(1,2)_6,\ \ if\ \ w\in W_5.\nonumber\\
\end{array}
\right.
$
\rm{(b)}\ \ \ $
|W_3\cap (W_0+w)|=\left\{
\begin{array}{llll}
(0,3)_6,\ \ if\ \ w\in W_0\nonumber\\
(5,2)_6,\ \ if\ \ w\in W_1\nonumber \\
(4,1)_6,\ \ if\ \ w\in W_2\nonumber\\
(3,0)_6,\ \ if\ \ w\in W_3\nonumber\\
(2,5)_6,\ \ if\ \ w\in W_4\nonumber\\
(1,4)_6,\ \ if\ \ w\in W_5.\nonumber\\
\end{array}
\right.\\
$
(c)\ \ \ $
|W_5\cap (W_0+w)|=\left\{
\begin{array}{llll}
(0,5)_6,\ \ if\ \ w\in W_0\nonumber\\
(5,4)_6,\ \ if\ \ w\in W_1\nonumber \\
(4,3)_6,\ \ if\ \ w\in W_2\nonumber\\
(3,2)_6,\ \ if\ \ w\in W_3\nonumber\\
(2,1)_6,\ \ if\ \ w\in W_4\nonumber\\
(1,0)_6,\ \ if\ \ w\in W_5.\nonumber\\
\end{array}
\right.
$
(d)\ \ \ $
|W_1\cap (W_2+w)|=\left\{
\begin{array}{llll}
(2,1)_6,\ \ if\ \ w\in W_0\nonumber\\
(1,0)_6,\ \ if\ \ w\in W_1\nonumber \\
(0,5)_6,\ \ if\ \ w\in W_2\nonumber\\
(5,4)_6,\ \ if\ \ w\in W_3\nonumber\\
(4,3)_6,\ \ if\ \ w\in W_4\nonumber\\
(3,2)_6,\ \ if\ \ w\in W_5.\nonumber\\
\end{array}
\right.\\
$
(e)\ \ \ $
|W_3\cap (W_2+w)|=\left\{
\begin{array}{llll}
(2,3)_6,\ \ if\ \ w\in W_0\nonumber\\
(1,2)_6,\ \ if\ \ w\in W_1\nonumber \\
(0,1)_6,\ \ if\ \ w\in W_2\nonumber\\
(5,0)_6,\ \ if\ \ w\in W_3\nonumber\\
(4,5)_6,\ \ if\ \ w\in W_4\nonumber\\
(3,4)_6,\ \ if\ \ w\in W_5.\nonumber\\
\end{array}
\right.
$
(f)\ \ \ $
|W_5\cap (W_2+w)|=\left\{
\begin{array}{llll}
(2,5)_6,\ \ if\ \ w\in W_0\nonumber\\
(1,4)_6,\ \ if\ \ w\in W_1\nonumber \\
(0,3)_6,\ \ if\ \ w\in W_2\nonumber\\
(5,2)_6,\ \ if\ \ w\in W_3\nonumber\\
(4,1)_6,\ \ if\ \ w\in W_4\nonumber\\
(3,0)_6,\ \ if\ \ w\in W_5.\nonumber\\
\end{array}
\right.\\
$
(g)\ \ \ $
|W_1\cap (W_4+w)|=\left\{
\begin{array}{llll}
(4,1)_6,\ \ if\ \ w\in W_0\nonumber\\
(3,0)_6,\ \ if\ \ w\in W_1\nonumber \\
(2,5)_6,\ \ if\ \ w\in W_2\nonumber\\
(1,4)_6,\ \ if\ \ w\in W_3\nonumber\\
(0,3)_6,\ \ if\ \ w\in W_4\nonumber\\
(5,2)_6,\ \ if\ \ w\in W_5.\nonumber\\
\end{array}
\right.
$
(h)\ \ \ $
|W_3\cap (W_4+w)|=\left\{
\begin{array}{llll}
(4,3)_6,\ \ if\ \ w\in W_0\nonumber\\
(3,2)_6,\ \ if\ \ w\in W_1\nonumber \\
(2,1)_6,\ \ if\ \ w\in W_2\nonumber\\
(1,0)_6,\ \ if\ \ w\in W_3\nonumber\\
(0,5)_6,\ \ if\ \ w\in W_4\nonumber\\
(5,4)_6,\ \ if\ \ w\in W_5.\nonumber\\
\end{array}
\right.\\
$
(i)\ \ \ $
|W_5\cap (W_4+w)|=\left\{
\begin{array}{llll}
(4,5)_6,\ \ if\ \ w\in W_0\nonumber\\
(3,4)_6,\ \ if\ \ w\in W_1\nonumber \\
(2,3)_6,\ \ if\ \ w\in W_2\nonumber\\
(1,2)_6,\ \ if\ \ w\in W_3\nonumber\\
(0,1)_6,\ \ if\ \ w\in W_4\nonumber\\
(5,0)_6,\ \ if\ \ w\in W_5.\nonumber
\end{array}
\right.
$\\
From the above discussion and using Lemma \ref{odd&even}, we get the value of $|(W_1\cup W_3\cup W_5)\cap ((W_0\cup W_2\cup W_4)+w)|$.
\end{proof}

\begin{lemma}\label{val p} 
Let the notations be same as before.\\
{\rm{(i)}} If $(n_1-1)(n_2-1)/36$ is even, then
\begin{align}
 |P\cap ((W_0\cup W_2\cup W_4)+w)|=\left\{ 
 \begin{array}{llll}
  0,\ \ \ \ \ \ \ \ \ \ \ \ \ \ \ \ \ \ \ if\ \ w\in P\nonumber\\
  (n_2-1)/2,\ \ \ \ \ \ \   if\ \ w\in Q\cup W_0\cup W_2\cup W_4\nonumber\\
  (n_2-1)/2-1,\ \ if\ \ w\in W_1\cup W_3\cup W_5.\nonumber
 \end{array}
 \right.
\end{align}
{\rm{(ii)}} If $(n_1-1)(n_2-1)/36$ is odd, then
\begin{align}
|P\cap ((W_0\cup W_2\cup W_4)+w)|=\left\{ 
 \begin{array}{ll}
  0,\ \ \ \ \ \ \ \ \ \ \ \ \ \ \ \ \ \ \ if\ \ w\in P\nonumber\\
  (n_2-1)/2,\ \ \ \ \ \ \ if\ \ w\in Q\cup W_1\cup W_3\cup W_5\nonumber\\
  (n_2-1)/2-1,\ \ if\ \ w\in  W_0\cup W_2\cup W_4.\nonumber\\
  \end{array}
 \right.
\end{align}
\end{lemma}

\begin{proof}
 We have $|P\cap ((W_0\cup W_2\cup W_4)+w)|= |P\cap (W_0+w)|+|P\cap (W_2+w)|+|P\cap (W_4+w)|$ and $|(P\cup R)\cap (W_i+w)|= |P\cap (W_i+w)|+|R\cap (W_i+w)|.$ By Lemma 3 and Corollary 1 of \cite{Auto13}, we have 
 \begin{align}
 |(P\cup R)\cap (W_i+w)|=\left\{ 
 \begin{array}{lll}
 0,\ \ \ \ \ \ \ \ if\ \ w \in P\cup R\nonumber\\
 \frac{n_2-1}{6},\ \ otherwise.
  \end{array}
  \right.
  \end{align}
  and
  \begin{align}
  |R\cap (W_i+w)|=\left\{
  \begin{array}{lll}
   1,\ \ \ \ if \ \ w \in W_i\ \ \ \ \mathrm{and}\ \ (n_1-1)(n_2-1)/36 \ \  \mathrm{is\ \ odd} \nonumber\\
   1,\ \ \ \ if \ \ w \in W_{i+3}\  \mathrm{and}\ \ (n_1-1)(n_2-1)/36\ \ \mathrm{is\ \ even} \nonumber\\
   0,\ \ \ otherwise. 
   \end{array}
\right.
  \end{align}
  
  Clearly, the Lemma follows from the above discussion.
\end{proof}

\begin{lemma}\label{1}
 Let $w\in \mathbb{Z}_n$, then we have
 \begin{align}
 |(W_1\cup W_3\cup W_5)\cap ((R\cup Q)+w)|=\left\{
 \begin{array}{ll}
  0,\ \ \ \ \ \ \ if \ \ w\in Q\cup R\\
  \frac{n_1-1}{2},\ \ otherwise.\nonumber
 \end{array}
 \right.
\end{align}
\end{lemma}

\begin{proof}
 By Lemma 4 of \cite{auto2}, we have
 \begin{align}
 |W_i\cap ((R\cup Q)+w)|=\left\{
 \begin{array}{ll}
  0,\ \ \ \ \ \ \ if \ \ w\in Q\cup R\\
  \frac{n_1-1}{6},\ \ otherwise.\nonumber
 \end{array}
 \right.
 \end{align}
This implies the lemma.
 \end{proof}

We recall the following Lemma from \cite{auto2}.

\begin{lemma}\label{2}
\cite[Lemma 8]{auto2} If $w\in \mathbb{Z}_n,$ then we have\\
\begin{align}
 |P\cap ((Q\cup R)+w)|=\left\{
 \begin{array}{lll}
 1,\ \ if \ \ w\in P\\
 0,\ \ if \ \ w\in Q\\
 1,\ \ if \ \ w\in \mathbb{Z}^*_n.\nonumber
 \end{array}
 \right.
\end{align}
\end{lemma}
 
The following theorems compute the autocorrelation values of two-prime WGCS-I.
\begin{theorem}\label{even}
 Let $(n_1-1)(n_2-1)/36$ be even. Then\\
 \begin{align}
  C_\lambda(w)=\left\{
  \begin{array}{lll}
  \frac{n_2-n_1-3}{n},\ \ if\ \ w\in P\\
  \frac{n_1-n_2+1}{n},\ \ if \ \ w\in Q\\
  -\frac{1}{n},\ \ \ \ \ \ \ \  if \ \ w\in \mathbb{Z}^*_n.\nonumber
  \end{array}
  \right.
 \end{align}
\end{theorem}

\begin{proof}
 Substituting the values from Lemma \ref{val}-\ref{2} in Eq.(\ref{ds}), we obtain 
 \begin{align}
d_\lambda(1,0;w) =\left\{
 \begin{array}{lll}
  9\frac{(n_1-1)(n_2-1)}{36}+0+\frac{n_1-1}{2}+1,\ \ \ \ \ \ \ \ \ \ \ \ \ \ if\ \ w\in P \nonumber\\
  9\frac{(n_1-1)(n_2-1)}{36}+\frac{n_2-1}{2}+0+0,\ \ \ \ \ \ \ \ \ \ \ \ \ \ if\ \ w\in Q \nonumber\\
   \frac{(n_1-2)(n_2-2)-3}{4}+\frac{n_2-1}{2}+\frac{n_1-1}{2}+1,\ \ \ \ \ \ \ if\ \ w\in W_0\cup W_2\cup W_4\nonumber\\
 \frac{(n_1-2)(n_2-2)+1}{4}+\frac{n_2-1}{2}-1+\frac{n_1-1}{2}+1,\ \ if\ \ w\in W_1\cup W_3\cup W_5.\nonumber
 \end{array}
 \right.\\
 =\left\{
 \begin{array}{lll}
  \frac{n_1n_2+n_1-n_2+3}{4},\ \ if\ \ w\in P\nonumber\\
  \frac{n_1n_2-n_1+n_2-1}{4},\ \ if\ \ w\in Q\nonumber\\
   \frac{n_1n_2+1}{4},\ \ \ \ \ \ \ \ \ \ \   if\ \ w\in \mathbb{Z}_n.\nonumber\\
   \end{array}
 \right.\ \ \ \ \ \ \ \ \ \ \ \ \ \ \ \ \ \ \ \ \ \ \ \ \ \ \ \ \ \ \ \ \ \ \ \ \ \ \ \ \ \ \ \ \ \ \ 
\end{align}
We get the value of $C_\lambda(w)$ by putting the value of $d_\lambda(1,0;w)$ in Lemma \ref{val d1}.
\end{proof}

\begin{theorem}\label{odd}
 Let $(n_1-1)(n_2-1)/36$ be odd. Then
 \begin{align}
  C_\lambda(w)=\left\{
  \begin{array}{lll}
  \frac{n_2-n_1-3}{n},\ \ if\ \ w\in P\\
  \frac{n_1+1-n_2}{n},\ \ if \ \ w\in Q\\
  \frac{1}{n},\ \ \ \ \ \ \ \ \ \    if \ \ w\in W_0\cup W_2\cup W_4\\ 
  -\frac{3}{n},\ \ \ \ \ \ \ \    if\ \ w\in W_1\cup W_3\cup W_5.\nonumber
  \end{array}
  \right.
 \end{align}
\end{theorem}

\begin{proof}
 The proof is similar to the prove of Theorem \ref{even}.
\end{proof}

By Theorem \ref{even}, the autocorrelation values $C_\lambda(w)$ is four-valued if $\frac{(n_1-1)(n_2-1)}{36}$ is even and by Theorem \ref{odd}, $C_\lambda(w)$ is five-valued if $\frac{(n_1-1)(n_2-1)}{36}$ is odd for any $|n_2-n_1|$ multiple of 6.

Now we employ the sequence $\lambda^\infty$ (defined in  Eq.(\ref{si})) to construct cyclic codes over $\mathrm{GF}(q)$.

\section{A class of cyclic codes over $\mathrm{GF}(q)$ defined by two-prime WGCS-I}\label{subsection 3.3} 

We have $\mathrm{gcd}(n,q)=1$. Let $m$ be the order of $q$ modulo $n$. Then the field $\mathrm{GF}(q^{m})$ has a primitive $nth$ root of unity $\beta$. We define
\begin{align}
\Lambda(x)=\sum\limits_{i\in C_{1}}x^{i}=\left(\sum\limits_{i\in P}+\sum\limits_{i\in W_{1}}+\sum\limits_{i\in W_{3}}+\sum\limits_{i\in W_{5}}\right)x^{i}\in \mathrm{GF}(q)[x],\label{lambda}
\end{align}
Our main aim in this section is to find the generator polynomial
\begin{align}
 g_\lambda(x)=\frac{x^n-1}{\mathrm{gcd}(x^n-1,\Lambda(x))}\nonumber
\end{align}
of the cyclic code $C_\lambda$ defined by the sequence $\lambda^{\infty}.$ 
To compute the parameters of the cyclic code $C_\lambda$ defined by the sequence $\lambda^{\infty}$, we need to compute $\mathrm{gcd}(x^{n}-1,\Lambda(x))$. Since $\beta$ is a primitive $n$th root of unity, we need only to find such $t$'s that $\Lambda(\beta^{t})=0$, where $0\leq t\leq n-1$. To this end, we need number of auxiliary results.
We have 
\begin{align}
0=\beta^{n}-1=(\beta^{n_1})^{n_2}-1=(\beta^{n_1}-1)(1 + \beta^{n_1} + \beta^{2n_1} + \cdots + \beta^{(n_2-1)n_1}).\nonumber
\end{align}
It follows that
\begin{align}
 \beta^{n_1}+\beta^{2n_1}+\cdots+\beta^{(n_2-1)n_1}=-1, i.e., \sum\limits_{i\in P} \beta^{i} = -1.\label{P}
\end{align}
By symmetry we get
\begin{align}
\beta^{n_2} + \beta^{2n_2} + \cdots + \beta^{(n_1-1)n_2} = -1, i.e., \sum\limits_{i\in Q}\beta^{i}=-1.\label{Q}
\end{align}

\begin{lemma}\label{lemma 3.4}
Let the symbols be same as before. For $0\leq j\leq5$, we have
$$
\sum\limits_{i\in W_{j}}\beta^{it}=\left\{
\begin{array}{ll}
-\frac{n_{1}-1}{6}\ (\mathrm{mod}\ p),\ \mathrm{if}\ t\in P\\
-\frac{n_{2}-1}{6}\ (\mathrm{mod}\ p),\ \mathrm{if}\ t\in Q.
\end{array}
\right.
$$
\end{lemma}

\begin{proof}
Suppose that $t\in Q$. Since $g$ is a common primitive roots of $n_1$ and $n_2$ and the order of $g$ modulo $n$ is $e$, by the definition of $u$, we have
\begin{align}
W_{j}\ \mathrm{mod}\ n_{1}&=\{g^{s}u^{j}\ \mathrm{mod}\ n_{1}:\ s=0,1,2,\cdots,e-1\}\nonumber\\
&=\{g^{s+j}\ \mathrm{mod}\ n_{1}:\ s=0,1,2,\cdots,e-1\}\nonumber\\
&=\frac{n_{2}-1}{6}\ast\{1,2,\cdots,n_{1}-1\},\nonumber
\end{align}
where $\frac{n_{2}-1}{6}$ denotes the multiplicity of each element in the set $\{1,2,\cdots,n_{1}-1\}$.
 We can write $g^sx^j$ in the form 
\begin{align}
&1+k_{11}n_1,1+k_{12}n_1,\cdots,1+k_{1(n_2-1)/6}n_1,\nonumber \\
&2+k_{21}n_1,2+k_{22}n_1,\cdots,2+k_{2(n_2-1)/6}n_1,\nonumber \\
\vdots \nonumber \\
&n_1-1+k_{(n_1-1)1}n_1,n_1-1+k_{(n_1-1)2}n_1,\cdots,n_1-1+k_{(n_1-1)(n_2-1)/6}n_1.\label{PQ}
 \end{align}
where $k_{li}$ is an positive integer, $1\leq l\leq n_1-1$ and $1\leq i\leq (n_2-1)/6$.
Since $s$ ranges over $\{0, 1, \cdots\ , e-1\},$ we divides the set $W_j$ into $(n_2-1)/6$ subsets each of which contains ${n_1 -1}$
consecutive integers, i.e., $ \ g^{s+j}\ \mathrm{mod}\ n_{1}$ takes on each element of $\{1,2,\cdots,n_{1}-1\}$ exactly $\frac{n_{2}-1}{6}$
times. From Eq.(\ref{PQ}), it follows that 
if $t\in Q,$ we have $\beta^{(m + k_{li}n_1)t}=\beta^{mt}$, where $1\leq m \leq n_1-1.$  
 It follows from Eq.(\ref{Q}) that
$$\sum\limits_{i\in W_{j}}\beta^{it}=\left(\frac{n_{2}-1}{6}\right)\sum\limits_{j\in Q}\beta^{j}=-\frac{n_{2}-1}{6}\ (\mathrm{mod}\ p).$$
For $t\in P$, we can get the result by similar argument.
\end{proof}

\begin{lemma}\label{mod d}
 For any $r\in W_{i}$, we have $rW_{j}=W_{(i+j) (\mathrm{mod}\ d)}$, where $rW_{j}=\{rt\ |\ t \in W_{j}\}$.
\end{lemma}

\begin{proof}
We have $ W_{i} =\{g^{s}u^{i}:s=0,1,2,\cdots,e-1\},i=0,1,\cdots,d-1 $ and let $r = g^{s_1}u^{i}\in W_{i}.$ 
Then $rW_{j} =g^{s_1}u^{i}\{u^{j},g^{1}u^{j},\cdots,g^{e-1}u^{j}\}=\{g^{s_1}u^{i+j},g^{s_1+1}u^{i+j},\cdots,g^{s_1+e-1}u^{i+j}\}.$
Since $u\in \mathbb{Z}_{n}^{\ast}$, there must exist an integer $\upsilon$ with $0\leqslant \upsilon \leqslant e-1$ such that $u^d=g^\upsilon$, therefore, we must have $rW_{j}=W_{(i+j)(\mathrm{mod}\ d)}.$
\end{proof}

\begin{lemma}\label{value}
Let $D_0=W_0\cup W_2\cup W_4$ and $D_1=W_1\cup W_3\cup W_5$ and rest of the symbols be same as before. For all $t\in \mathbb{Z}_{n}$ we have
$$
\Lambda(\beta^{t})=\left\{
\begin{array}{llllll}
-\frac{n_{1}+1}{2}\ (\mathrm{mod}\ p),\ \mathrm{if}\ t\in P\\
\ \ \frac{n_{2}-1}{2}\ (\mathrm{mod}\ p),\ \mathrm{if}\ t\in Q\\
\ \ \Lambda(\beta),\ \ \ \ \ \ \ \ \ \ \ \ \mathrm{if}\ t\in D_{0}\\
-(\Lambda(\beta)+1),\ \ \ \ \ \mathrm{if}\ t\in D_{1}.
\end{array}
\right.
$$
\end{lemma}

\begin{proof}
Since $\mathrm{gcd}(n_{1},n_{2})=1$, we have $tP=P$ if $t\in P$. By Eqs.(\ref{lambda}), (\ref{P}) and Lemma \ref{lemma 3.4}, we get 
\begin{align}
\Lambda(\beta^{t})=\sum\limits_{i\in C_{1}}\beta^{ti}&=\left(\sum\limits_{i\in P}+\sum\limits_{i\in W_{1}}+\sum\limits_{i\in W_{3}}+\sum\limits_{i\in W_{5}}\right)\beta^{ti}\nonumber\\
&=(-1\ \mathrm{mod}\ p)  -\left(\frac{n_{1}-1}{6}\ \mathrm{mod}\ p\right)-\left(\frac{n_{1}-1}{6}\ \mathrm{mod}\ p\right)-\left(\frac{n_{1}-1}{6}\ \mathrm{mod}\ p\right)\nonumber\\
&=-\frac{n_{1}+1}{2}\ \mathrm{mod}\ p.\nonumber
\end{align}
If $t\in Q$, then $tP=0$. By Eqs.(\ref{lambda}), (\ref{P}) and Lemma \ref{lemma 3.4}, we get  
\begin{align}
\Lambda(\beta^{t})=\sum\limits_{i\in C_{1}}\beta^{ti}&=\left(\sum\limits_{i\in P}+\sum\limits_{i\in W_{1}}+\sum\limits_{i\in W_{3}}+\sum\limits_{i\in W_{5}}\right)\beta^{ti}\nonumber\\
&=(n_{2}-1 \ \mathrm{mod}\ p)-\left(\frac{n_{2}-1}{6}\ \mathrm{mod}\ p\right)-\left(\frac{n_{2}-1}{6}\ \mathrm{mod}\ p\right)-\left(\frac{n_{2}-1}{6}\ \mathrm{mod}\ p\right)\nonumber\\
&=\frac{n_{2}-1}{2}\ \mathrm{mod}\ p.\nonumber
\end{align}  
If $t\in D_0$, we have three cases:\\
Case I: Let $t\in W_0$, then by Lemma \ref{mod d}, we have $tW_i=W_i.$ Since $\mathrm{gcd}(t,n_2)=1,$ we have $tP=P$ if $t\in W_0.$ Hence
\begin{align}
\Lambda(\beta^{t})=\sum\limits_{i\in C_{1}}\beta^{ti}&=\left(\sum\limits_{i\in P}+\sum\limits_{i\in W_{1}}+\sum\limits_{i\in W_{3}}+\sum\limits_{i\in W_{5}}\right)\beta^{ti} \nonumber \\
&=\left(\sum\limits_{i\in P}+\sum\limits_{i\in W_{1}}+\sum\limits_{i\in W_{3}}+\sum\limits_{i\in W_{5}}\right)\beta^{i} \nonumber \\
&=\Lambda(\beta). \nonumber
\end{align}
Case II: Let $t\in W_2$, then by Lemma \ref{mod d}, we have $tW_i=W_{(i+2)(\mathrm{mod}\ 6)}$ for $0 \leqslant i \leqslant 5$. Since $\mathrm{gcd}(t,n_2)=1,$ we have $tP=P$ if $t\in W_2$. Hence
\begin{align}
\Lambda(\beta^{t})=\sum\limits_{i\in C_{1}}\beta^{ti}&=\left(\sum\limits_{i\in P}+\sum\limits_{i\in W_{1}}+\sum\limits_{i\in W_{3}}+\sum\limits_{i\in W_{5}}\right)\beta^{ti}\nonumber\\
&=\left(\sum\limits_{i\in P}+\sum\limits_{i\in W_{3}}+\sum\limits_{i\in W_{5}}+\sum\limits_{i\in W_{1}}\right)\beta^{i}\nonumber\\
&=\Lambda(\beta). \nonumber
\end{align}
Case III: Let $t\in W_4$, then by Lemma \ref{mod d}, we have $tW_i=W_{(i+4)(\mathrm{mod}\ 6)}$ for $0 \leqslant i \leqslant 5$. Since $\mathrm{gcd}(t,n_2)=1$, we have $tP=P$ if $t\in W_4$. Hence
\begin{align}
\Lambda(\beta^{t})=\sum\limits_{i\in C_{1}}\beta^{ti}&=\left(\sum\limits_{i\in P}+\sum\limits_{i\in W_{1}}+\sum\limits_{i\in W_{3}}+\sum\limits_{i\in W_{5}}\right)\beta^{ti}\nonumber\\
&=\left(\sum\limits_{i\in P}+\sum\limits_{i\in W_{5}}+\sum\limits_{i\in W_{1}}+\sum\limits_{i\in W_{3}}\right)\beta^{i}\nonumber\\
&=\Lambda(\beta). \nonumber
\end{align}
Similarly, if $t\in D_1$, we have three cases:\\
Case I: Let $t\in W_1$ then by Lemma \ref{mod d}, we have $tW_i=W_{i+1\ (\mathrm{mod}\ 6)}.$ Since $\mathrm{gcd}(t,n_2)=1$, we have $tP=P$ if $t\in W_1$.
We have $\beta^n-1=(\beta-1)(\sum\limits_{i=0}^{n-1}\beta^i)=0$ and $\beta-1\neq0,$ this give $\sum\limits_{i=0}^{n-1}\beta^i=0.$ Therefore, $\sum\limits_{i=0}^{n-1}\beta^i=1+\sum\limits_{i\in P}\beta^i+\sum\limits_{i\in Q}\beta^i+\sum\limits_{i\in\bigcup\limits_{j=0}^{5} W_{j}}\beta^{i}=0.$ 
From Eqs.(\ref{P}) and (\ref{Q}), we get
\begin{equation}
\sum\limits_{i\in\bigcup\limits_{j=0}^{5} W_{j}}\beta^{i}=1.\label{sum1}
\end{equation}
Hence
\begin{align}
\Lambda(\beta^{t})=\sum\limits_{i\in C_{1}}\beta^{ti}&=\left(\sum\limits_{i\in P}+\sum\limits_{i\in W_{1}}+\sum\limits_{i\in W_{3}}+\sum\limits_{i\in W_{5}}\right)\beta^{ti}\nonumber\\
&=\left(\sum\limits_{i\in P}+\sum\limits_{i\in W_{2}}+\sum\limits_{i\in W_{4}}+\sum\limits_{i\in W_{0}}\right)\beta^{i}\nonumber\\
&=\left(\sum\limits_{i\in P}-\sum\limits_{i\in W_{1}}-\sum\limits_{i\in W_{3}}-\sum\limits_{i\in W_{5}}\right)\beta^{i}+1\nonumber\\
&=\left(-\sum\limits_{i\in P}-\sum\limits_{i\in W_{1}}-\sum\limits_{i\in W_{3}}-\sum\limits_{i\in W_{5}}\right)\beta^{i}+2\sum\limits_{i\in P}\beta^{i}+1\nonumber
\end{align}
From Eq.(\ref{P}), we know that $\sum\limits_{i\in P}\beta^i=-1$ and by the definition of $\Lambda(\beta)$, we have $\Lambda(\beta^t)=-(\Lambda(\beta)+1)$.
Similarly, we can prove other two cases namely, Case II : $t \in W_3$ and Case III : $t \in W_5.$ 
\end{proof}

\begin{lemma}\label{lemma 3.7}
 If $q\in D_0,$ we have $\Lambda(\beta) \in \mathrm{GF}(q)$ and $(\Lambda(\beta))^q=\Lambda(\beta)$.
 If $q\in D_1,$ we have  $\Lambda(\beta)^q=-(\Lambda(\beta)+1)$.
\end{lemma}

\begin{proof} 
We have $\mathrm{gcd}(n,q)=1,\mathrm{i.e.,}\ q \in \mathbb{Z}^*_n, $ then $q\in\bigcup\limits_{i=1}^{5}W_i=D_0\cup D_1$. If $q\in D_{0}$, by Lemma \ref{value}, we have $(\Lambda(\beta))^{q}=\Lambda(\beta^{q})=\Lambda(\beta)$. So, $\Lambda(\beta) \in \mathrm{GF}(q).$ Similarly, if
 $q \in D_1$, from Lemma \ref{value}, the result follows.
\end{proof}

\begin{lemma}\label{imp}
 If $n_1n_2\equiv 1 \ (\mathrm{mod} \ 12 )$, we have 
 \begin{equation}
  \Lambda(\beta)(\Lambda(\beta)+1)=\frac{n-1}{4}.\nonumber\\
 \end{equation}
If $n_1n_2\equiv 7 \ (\mathrm{mod} \ 12 )$, we have 
 \begin{equation}
  \Lambda(\beta)(\Lambda(\beta)+1)=-\frac{n+1}{4}.\nonumber\\
 \end{equation}
 \end{lemma}

\begin{proof} 
 We have 
 \begin{align}
 \Lambda(\beta)=-1+\sum\limits_{i\in W_{1}}\beta^{i}+\sum\limits_{i\in W_{3}}\beta^{i}+\sum\limits_{i\in W_{5}}\beta^{i},\nonumber
\end{align}
and
\begin{align}
\Lambda(\beta)(\Lambda(\beta)+1)=-\left(\sum\limits_{i\in W_{1}}\beta^{i}+\sum\limits_{i\in W_{3}}\beta^{i}+\sum\limits_{i\in W_{5}}\beta^{i}\right)+\sum\limits_{i\in W_{1}}\sum\limits_{j\in W_{1}}\beta^{i+j}+\sum\limits_{i\in W_{3}}\sum\limits_{j\in W_{3}}\beta^{i+j}+\sum\limits_{i\in W_{5}}\sum\limits_{j\in W_{5}}\beta^{i+j}\nonumber\\
+ 2\sum\limits_{i\in W_{1}}\sum\limits_{j\in W_{3}}\beta^{i+j}+ 2\sum\limits_{i\in W_{3}}\sum\limits_{j\in W_{5}}\beta^{i+j}+ 2\sum\limits_{i\in W_{5}}\sum\limits_{j\in W_{1}}\beta^{i+j}.\label{sumo}  \ \ \ \ \ \ \ \ \ \ \ \ \ \ \ \ \ \ \ \ \ \ \ \ \ \ \ \ \ \ \ \ \ \ \  
\end{align}
Let $n_1n_2\equiv 1 \ (\mathrm{mod} \ 12 )$ from Lemma \ref{W0}, $-1\in W_0$ and from Lemma \ref{mod d}, $-W_j=\{ -t : t\in W_j\}=W_j.$ 
\begin{align}
&\ \ \ \ \  \sum\limits_{i\in W_{1}}\sum\limits_{j\in W_{1}}\beta^{i+j}=\sum\limits_{i\in W_{1}}\sum\limits_{j\in W_{1}}\beta^{i-j}\nonumber  \\
&\ \ \ \ \ \ =|W_1|+\sum\limits_{r\in P\cup Q}d(1,1;r)\beta^{r}+(1,1)_6\sum\limits_{i\in W_{0}}\beta^{i}+(0,0)_6\sum\limits_{i\in W_{1}}\beta^{i}+(5,5)_6\sum\limits_{i\in W_{2}}\beta^{i}+(4,4)_6\sum\limits_{i\in W_{3}}\beta^{i}\nonumber \\
&\ \ \ \ \ \ \ \ \ +(3,3)_6\sum\limits_{i\in W_{4}}\beta^{i}+(2,2)_6\sum\limits_{i\in W_{5}}\beta^{i},\label{odd1}
\end{align}
\begin{align}
&\ \ \ \ \  \sum\limits_{i\in W_{3}}\sum\limits_{j\in W_{3}}\beta^{i+j}=\sum\limits_{i\in W_{3}}\sum\limits_{j\in W_{3}}\beta^{i-j}\nonumber \\
&\ \ \ \ \ \ =|W_3|+\sum\limits_{r\in P\cup Q}d(3,3;r)\beta^{r}+(3,3)_6\sum\limits_{i\in W_{0}}\beta^{i}+(2,2)_6\sum\limits_{i\in W_{1}}\beta^{i}+(1,1)_6\sum\limits_{i\in W_{2}}\beta^{i}+(0,0)_6\sum\limits_{i\in W_{3}}\beta^{i}\nonumber \\
&\ \ \ \ \ \ \ \ \ +(5,5)_6\sum\limits_{i\in W_{4}}\beta^{i}+(4,4)_6\sum\limits_{i\in W_{5}}\beta^{i},\label{odd2}
\end{align}
\begin{align}
&\ \ \ \ \  \sum\limits_{i\in W_{5}}\sum\limits_{j\in W_{5}}\beta^{i+j}=\sum\limits_{i\in W_{5}}\sum\limits_{j\in W_{5}}\beta^{i-j}\nonumber \\
&\ \ \ \ \ \ =|W_5|+\sum\limits_{r\in P\cup Q}d(5,5;r)\beta^{r}+(5,5)_6\sum\limits_{i\in W_{0}}\beta^{i}+(4,4)_6\sum\limits_{i\in W_{1}}\beta^{i}+(3,3)_6\sum\limits_{i\in W_{2}}\beta^{i}+(2,2)_6\sum\limits_{i\in W_{3}}\beta^{i}\nonumber \\
&\ \ \ \ \ \ \ \ \ +(1,1)_6\sum\limits_{i\in W_{4}}\beta^{i}+(0,0)_6\sum\limits_{i\in W_{5}}\beta^{i},\label{odd3}
\end{align}
\begin{align}
&2 \sum\limits_{i\in W_{1}}\sum\limits_{j\in W_{3}}\beta^{i+j}=2\sum\limits_{i\in W_{1}}\sum\limits_{j\in W_{3}}\beta^{i-j}\nonumber \\
&= 2\ \left(\sum\limits_{r\in P\cup Q}d(3,1;r)\beta^{r}+(3,1)_6\sum\limits_{i\in W_{0}}\beta^{i}+(2,0)_6\sum\limits_{i\in W_{1}}\beta^{i}+(1,5)_6\sum\limits_{i\in W_{2}}\beta^{i}+(0,4)_6\sum\limits_{i\in W_{3}}\beta^{i}\right.\nonumber \\
&\ \ \ \left.+(5,3)_6\sum\limits_{i\in W_{4}}\beta^{i}+(4,2)_6\sum\limits_{i\in W_{5}}\beta^{i} \right) ,\label{odd4}
\end{align}
\begin{align}
&2 \sum\limits_{i\in W_{3}}\sum\limits_{j\in W_{5}}\beta^{i+j}=2\sum\limits_{i\in W_{3}}\sum\limits_{j\in W_{5}}\beta^{i-j}\nonumber \\
&=2 \left(\sum\limits_{r\in P\cup Q}d(5,3;r)\beta^{r}+(5,3)_6\sum\limits_{i\in W_{0}}\beta^{i}+(4,2)_6\sum\limits_{i\in W_{1}}\beta^{i}+(3,1)_6\sum\limits_{i\in W_{2}}\beta^{i}+(2,0)_6\sum\limits_{i\in W_{3}}\beta^{i}\right.\nonumber \\
&\ \ \ \left.+(1,5)_6\sum\limits_{i\in W_{4}}\beta^{i}+(0,4)_6\sum\limits_{i\in W_{5}}\beta^{i} \right) ,\label{odd5}
\end{align}
\begin{align}
&2 \sum\limits_{i\in W_{5}}\sum\limits_{j\in W_{1}}\beta^{i+j}=2\sum\limits_{i\in W_{5}}\sum\limits_{j\in W_{1}}\beta^{i-j}\nonumber \\
&= 2 \left(\sum\limits_{r\in P\cup Q}d(1,5;r)\beta^{r}+(1,5)_6\sum\limits_{i\in W_{0}}\beta^{i}+(0,4)_6\sum\limits_{i\in W_{1}}\beta^{i}+(5,3)_6\sum\limits_{i\in W_{2}}\beta^{i}+(4,2)_6\sum\limits_{i\in W_{3}}\beta^{i}\right.\nonumber \\
&\ \ \ \left.+(3,1)_6\sum\limits_{i\in W_{4}}\beta^{i}+(2,0)_6\sum\limits_{i\in W_{5}}\beta^{i}\right),\label{odd6}
\end{align}
Substituting the values of Eqs.(\ref{odd1})-(\ref{odd6}) into Eq.(\ref{sumo}) and then from Lemma \ref{val d} and \ref{odd&even}, and Eq.(\ref{sum1}), we get
\begin{align}
\Lambda(\beta)(\Lambda(\beta)+1)=-\left(\sum\limits_{i\in W_{1}}\beta^{i}+\sum\limits_{i\in W_{3}}\beta^{i}+\sum\limits_{i\in W_{5}}\beta^{i}\right) + \left(\frac{3M}{2}\right)\sum\limits_{i\in W_0}\beta^i\nonumber
+ \left(\frac{3M}{2}+1\right)\sum\limits_{i\in W_1}\beta^i \ \ \ \ \ \ \ \ \ \ \ \ \ \ \ \ \ \ \ \ \ \ \ \ \ \ \ \ \ \ \ \ \  \ \ \ \ \ \\ \nonumber 
  +  \left(\frac{3M}{2}\right)\sum\limits_{i\in W_2}\beta^i + \left(\frac{3M}{2}+1\right)\sum\limits_{i\in W_3}\beta^i  + \left(\frac{3M}{2}\right)\sum\limits_{i\in W_4}\beta^i + \left(\frac{3M}{2}+1\right)\sum\limits_{i\in W_5}\beta^i \ \ \ \ \ \ \ \ \ \ \ \ \ \ \ \ \ \ \ \ \ \ \ \ \ \ \ \ \ \ \\ \nonumber
-12\frac{(n_1-1)(n_2-1)}{36}-3\frac{(n_1-1)(n_2-7)}{36} -3\frac{(n_1-7)(n_2-1)}{36}+3\frac{(n_1-1)(n_2-1)}{6}\ \ \ \ \ \ \ \ \ \ \ \ \ \ \ \ \ \ \ \ \ \ \ \ \  \\ 
=\frac{n-1}{4}.\ \ \ \ \ \ \ \ \ \ \ \ \ \ \ \ \ \ \ \ \ \ \ \ \ \ \ \ \ \ \ \ \ \ \ \ \ \ \ \ \ \ \ \ \ \ \ \ \ \ \ \ \ \ \ \ \ \ \ \ \ \ \ \ \ \ \ \ \ \ \ \ \ \ \ \ \ \ \ \ \ \ \ \ \ \ \ \ \ \ \ \ \ \ \ \ \ \ \ \ \ \ \ \ \ \ \ \ \ \ \ \ \ \  \ \ \nonumber
\end{align}
It completes the first part of the Lemma.\\
Now suppose that  $n_1n_2\equiv 7 \ (\mathrm{mod} \ 12 ).$ By Lemma \ref{W3}, $-1\in W_3$ and from Lemma \ref{mod d}, $-W_j=\{ -t : t\in W_j\}=W_{(j+3)  (\mathrm{mod}\ 6)}.$ 
\begin{align}
&\ \ \ \ \  \sum\limits_{i\in W_{1}}\sum\limits_{j\in W_{1}}\beta^{i+j}=\sum\limits_{i\in W_{1}}\sum\limits_{j\in W_{4}}\beta^{i-j}\nonumber  \\
&\ \ \ \ \ \ =\sum\limits_{r\in P\cup Q}d(4,1;r)\beta^{r}+(4,1)_6\sum\limits_{i\in W_{0}}\beta^{i}+(3,0)_6\sum\limits_{i\in W_{1}}\beta^{i}+(2,5)_6\sum\limits_{i\in W_{2}}\beta^{i}+(1,4)_6\sum\limits_{i\in W_{3}}\beta^{i}\nonumber \\
&\ \ \ \ \ \ \ \ \ +(0,3)_6\sum\limits_{i\in W_{4}}\beta^{i}+(5,2)_6\sum\limits_{i\in W_{5}}\beta^{i},\label{even1}
\end{align}
\begin{align}
&\ \ \ \ \  \sum\limits_{i\in W_{3}}\sum\limits_{j\in W_{3}}\beta^{i+j}=\sum\limits_{i\in W_{3}}\sum\limits_{j\in W_{0}}\beta^{i-j}\nonumber \\
&\ \ \ \ \ \ =\sum\limits_{r\in P\cup Q}d(0,3;r)\beta^{r}+(0,3)_6\sum\limits_{i\in W_{0}}\beta^{i}+(5,2)_6\sum\limits_{i\in W_{1}}\beta^{i}+(4,1)_6\sum\limits_{i\in W_{2}}\beta^{i}+(3,0)_6\sum\limits_{i\in W_{3}}\beta^{i}\nonumber \\
&\ \ \ \ \ \ \ \ \ +(2,5)_6\sum\limits_{i\in W_{4}}\beta^{i}+(1,4)_6\sum\limits_{i\in W_{5}}\beta^{i},\label{even2}
\end{align}
\begin{align}
&\ \ \ \ \  \sum\limits_{i\in W_{5}}\sum\limits_{j\in W_{5}}\beta^{i+j}=\sum\limits_{i\in W_{5}}\sum\limits_{j\in W_{2}}\beta^{i-j}\nonumber \\
&\ \ \ \ \ \ =\sum\limits_{r\in P\cup Q}d(2,5;r)\beta^{r}+(2,5)_6\sum\limits_{i\in W_{0}}\beta^{i}+(1,4)_6\sum\limits_{i\in W_{1}}\beta^{i}+(0,3)_6\sum\limits_{i\in W_{2}}\beta^{i}+(5,2)_6\sum\limits_{i\in W_{3}}\beta^{i}\nonumber \\
&\ \ \ \ \ \ \ \ \ +(4,1)_6\sum\limits_{i\in W_{4}}\beta^{i}+(3,0)_6\sum\limits_{i\in W_{5}}\beta^{i},\label{even3}
\end{align}
\begin{align}
&2 \sum\limits_{i\in W_{1}}\sum\limits_{j\in W_{3}}\beta^{i+j}=2\sum\limits_{i\in W_{1}}\sum\limits_{j\in W_{0}}\beta^{i-j}\nonumber \\
&= 2\ \left(\sum\limits_{r\in P\cup Q}d(0,1;r)\beta^{r}+(0,1)_6\sum\limits_{i\in W_{0}}\beta^{i}+(5,0)_6\sum\limits_{i\in W_{1}}\beta^{i}+(4,5)_6\sum\limits_{i\in W_{2}}\beta^{i}+(3,4)_6\sum\limits_{i\in W_{3}}\beta^{i}\right.\nonumber \\
&\ \ \ \left.+(2,3)_6\sum\limits_{i\in W_{4}}\beta^{i}+(1,2)_6\sum\limits_{i\in W_{5}}\beta^{i} \right),\label{even4}
\end{align}
\begin{align}
&2 \sum\limits_{i\in W_{3}}\sum\limits_{j\in W_{5}}\beta^{i+j}=2\sum\limits_{i\in W_{3}}\sum\limits_{j\in W_{2}}\beta^{i-j}\nonumber \\
&=2 \left(\sum\limits_{r\in P\cup Q}d(2,3;r)\beta^{r}+(2,3)_6\sum\limits_{i\in W_{0}}\beta^{i}+(1,2)_6\sum\limits_{i\in W_{1}}\beta^{i}+(0,1)_6\sum\limits_{i\in W_{2}}\beta^{i}+(5,0)_6\sum\limits_{i\in W_{3}}\beta^{i}\right.\nonumber \\
&\ \ \ \left.+(4,5)_6\sum\limits_{i\in W_{4}}\beta^{i}+(3,4)_6\sum\limits_{i\in W_{5}}\beta^{i} \right),\label{even5}
\end{align}
\begin{align}
&2 \sum\limits_{i\in W_{5}}\sum\limits_{j\in W_{1}}\beta^{i+j}=2\sum\limits_{i\in W_{5}}\sum\limits_{j\in W_{4}}\beta^{i-j}\nonumber \\
&= 2 \left(\sum\limits_{r\in P\cup Q}d(4,5;r)\beta^{r}+(4,5)_6\sum\limits_{i\in W_{0}}\beta^{i}+(3,4)_6\sum\limits_{i\in W_{1}}\beta^{i}+(2,3)_6\sum\limits_{i\in W_{2}}\beta^{i}+(1,2)_6\sum\limits_{i\in W_{3}}\beta^{i}\right.\nonumber \\
&\ \ \ \left.+(0,1)_6\sum\limits_{i\in W_{4}}\beta^{i}+(5,0)_6\sum\limits_{i\in W_{5}}\beta^{i}\right),\label{even6}
\end{align}
Substituting the value from Eqs.(\ref{even1}) - (\ref{even6}) into Eq.(\ref{sumo}) and then from Lemma \ref{val d} and \ref{odd&even}, and Eq.(\ref{sum1}), we get
\begin{align}
\Lambda(\beta)(\Lambda(\beta)+1)=-\left(\sum\limits_{i\in W_{1}}\beta^{i}+\sum\limits_{i\in W_{3}}\beta^{i}+\sum\limits_{i\in W_{5}}\beta^{i}\right) + \left(\frac{3M}{2}-\frac{1}{2}\right)\sum\limits_{i\in W_0}\beta^i\nonumber
+ \left(\frac{3M}{2}+\frac{1}{2}\right)\sum\limits_{i\in W_1}\beta^i \ \ \ \ \ \ \ \ \ \ \ \ \ \ \ \ \ \ \ \ \ \ \ \ \ \ \ \ \ \ \ \ \  \ \ \ \ \ \\ \nonumber 
  +  \left(\frac{3M}{2}-\frac{1}{2}\right)\sum\limits_{i\in W_2}\beta^i + \left(\frac{3M}{2}+\frac{1}{2}\right)\sum\limits_{i\in W_3}\beta^i  + \left(\frac{3M}{2}-\frac{1}{2}\right)\sum\limits_{i\in W_4}\beta^i \ \ \ \ \ \ \ \ \ \ \ \ \ \  \ \ \ \ \ \ \  \ \ \ \ \ \ \ \ \ \ \ \ \ \ \ \ \ \ \ \ \ \ \ \ \ \ \ \ \ \ \ \\ \nonumber
 + \left(\frac{3M}{2}+\frac{1}{2}\right)\sum\limits_{i\in W_5}\beta^i - 18\frac{(n_1-1)(n_2-1)}{36} \ \ \ \ \ \ \ \ \ \ \ \ \  \  \ \ \ \ \ \ \ \ \ \ \ \ \ \ \ \ \ \ \ \ \ \ \ \ \ \ \ \  \ \ \ \ \ \ \ \ \ \ \ \  \ \ \ \ \ \ \ \ \ \ \ \ \ \ \ \ \ \ \ \ \ \ \ \ \ \ \ \  \\ 
=-\frac{n+1}{4}.\ \ \ \ \ \ \ \ \ \ \ \ \ \ \ \ \ \ \ \ \ \ \ \ \ \ \ \ \ \ \ \ \ \ \ \ \ \ \ \ \ \ \ \ \ \ \ \ \ \ \ \ \ \ \ \ \ \ \ \ \ \ \ \ \ \ \ \ \ \ \ \ \ \ \ \ \ \ \ \ \ \ \ \ \ \ \ \ \ \ \ \ \ \ \ \ \ \ \ \ \ \ \ \ \ \ \ \ \ \ \ \ \ \ \ \ \ \ \ \ \ \ \ \ \  \ \ \nonumber
\end{align}
This completes the proof of the Lemma. 
\end{proof}

Note that 
\begin{align}
\Lambda(1)=\frac{(n_1+1)(n_2-1)}{2}\ (\mathrm{mod}\ p).\label{lam1}
\end{align}

It is elementary to prove the following Lemma:
\begin{lemma}\label{quad}
 If $p$ is an odd prime, then
 \begin{align}
  \left(\frac{2}{p}\right)=\left\{
  \begin{array}{ll}
  1, \ \ \ \ if\ \ p\equiv 1 \ (\mathrm{mod}\ 24)\ or\ p\equiv 7 \ (\mathrm{mod}\ 24)\\
 -1, \ \ if\ \ p\equiv 13 \ (\mathrm{mod}\ 24)\ or\ p\equiv 19 \ (\mathrm{mod}\ 24).\nonumber
  \end{array}
\right.
 \end{align}
\end{lemma}

\begin{lemma}\label{D0}
  If $n\equiv 7\ (\mathrm{mod}\ 12)$ and $\frac{n+1}{4}\equiv 0 \ (\mathrm{mod}\ p)$ or $n\equiv 1\ (\mathrm{mod}\ 12)$ and $\frac{n-1}{4}\equiv 0\ (\mathrm{mod}\ p),$ then $q\ (\mathrm{mod}\ n)\in D_0.$
 \end{lemma}

 \begin{proof} 
 First, we prove that if $n\equiv 7\ (\mathrm{mod}\ 12)$ and $\frac{n+1}{4}\equiv 0 \ (\mathrm{mod}\ p)$, then $q \ (\mathrm{mod}\ n)\in D_0.$
 Clearly, $D_0$ is a subgroup of $\mathbb{Z}_n^*$. Hence, $D_0$ is a multiplicative group. Since $q$ is a power of $p$, it is sufficient to prove that $p \in D_0$.
 Suppose on the contrary that $p\in D_1.$ We first consider the case for $p=2.$ Let $2 \in D_1$. By the definition 
  of Whiteman's generalized cyclotomic classes, $2=u^sg^i, 0\leqslant i\leqslant e-1 $ and $s$ is odd. From (\ref{defu}), we have
  \begin{align}
   2\equiv g^{s+i} \ (\mathrm{mod}\ n_1)\ \ \ \ \ \mathrm{and}\ \ \ \ \ 2 \equiv g^{i} \ (\mathrm{mod}\ n_2).\nonumber
  \end{align}
  Therefore, $2$ must be a quadratic residue (non residue, respectively) modulo $n_1$ if it is a quadratic non residue (residue, respectively)
  modulo $n_2$.\\
  For $p=2$, if $\frac{n+1}{4}\equiv 0\ (\mathrm{mod}\ p)$ then $8$ divides $n_1n_2+1.$ Since 
  $\mathrm{gcd} (n_1-1,n_2-1)=6\ ,$ it is easy to check that we have only the following four possibilities:\\
  \begin{align}
  \left\{
\begin{array}{ll}
 n_1 \equiv 1 \ (\mathrm{mod}\ 24)\\
 n_2 \equiv 7 \ (\mathrm{mod}\ 24),
\end{array}
\right.
 \left\{
\begin{array}{ll}
 n_1 \equiv 7\ (\mathrm{mod}\ 24)\\
 n_2 \equiv 1 \ (\mathrm{mod}\ 24),
\end{array}
\right.
\left\{
\begin{array}{ll}
 n_1 \equiv 13 \ (\mathrm{mod}\ 24)\\
 n_2 \equiv 19 \ (\mathrm{mod}\ 24),
\end{array}
 \right.
 \left\{
\begin{array}{ll}
 n_1 \equiv 19 \ (\mathrm{mod}\ 24)\\
 n_2 \equiv 13 \ (\mathrm{mod}\ 24).\nonumber
\end{array}
\right.
\end{align}
By Lemma \ref{quad}, it follows that none of the above four possibilities are possible. This gives a contradiction therefore $2\in D_0.$\\
Let $p$ be an odd prime. Suppose on the contrary that $p\in D_1.$ By the definition, $p=u^sg^i, 0\leqslant i \leqslant e-1$ and $s$ is odd.
We have
\begin{align}
   p \equiv g^{s+i} \ (\mathrm{mod}\ n_1)\ \ \ \ \ \  \mathrm{and}\ \ \ \ \ \  p \equiv g^{i} \ (\mathrm{mod}\ n_2).\nonumber
  \end{align}
  Since $s$ is odd, then we must have
  \begin{align}
   \left(\frac{p}{n_1}\right)\left(\frac{p}{n_2}\right)=-1, \label{-1}
  \end{align}
where $(-)$ is the Legendre symbol.
If $n\equiv 7 \ (\mathrm{mod}\ 12)$, by Lemma \ref{W3}, $(n_1+n_2)/2$ is even. If $\frac{n+1}{4} \equiv 0\ (\mathrm{mod}\ p),$ then $n=n_1n_2\equiv -1\ (\mathrm{mod}\ p).$   
 By the Law of Quadratic Reciprocity,
  \begin{align}
  \left(\frac{p}{n_i}\right) = \left( -1 \right)^{\left(\frac{p-1}{2}\right)\left(\frac{n_i-1}{2}\right)}\left(\frac{n_i}{p}\right)\ \  \mathrm{for}\ \ i=1,2,\nonumber
  \end{align}
  and
  \begin{align}
   \left(\frac{-1}{p}\right)=\left(-1\right)^{\frac{p-1}{2}}.\nonumber
  \end{align}
  It follows that
  \begin{align}
    \left(\frac{p}{n_1}\right)\left(\frac{p}{n_2}\right)=\left(-1\right)^{\left(\frac{p-1}{2}\right)\left(\frac{n_1+n_2-2}{2}\right)}\left(\frac{n_1}{p}\right)\left(\frac{n_2}{p}\right)\nonumber\\
    =\left(-1\right)^{\left(\frac{p-1}{2}\right)\left(\frac{n_1+n_2-2}{2}\right)}\left(\frac{n_1n_2}{p}\right)\ \ \ \ \ \  \nonumber\\
    =\left(-1\right)^{\left(\frac{p-1}{2}\right)\left(\frac{n_1+n_2-2}{2}\right)}\left(\frac{-1}{p}\right)\ \ \ \ \  \ \ \  \nonumber\\
    =1.\ \ \ \ \ \ \ \ \ \ \ \ \ \ \ \ \ \ \ \ \ \ \ \ \  \ \ \ \ \ \ \ \ \ \ \ \ \  \ \ \ \nonumber
  \end{align}
  This is contrary to Eq.(\ref{-1}). Thus, $p\in D_0$. Similarly, we prove that if $n\equiv 1\ (\mathrm{mod}\ 12)$ and $\frac{n-1}{4}\equiv 0\ (\mathrm{mod}\ p),$ then $q\ (\mathrm{mod}\ n)\in D_0.$
  \end{proof}

We need to discuss the factorization of $x^n-1$ over $\mathrm{GF}(q).$
Let $\beta$ be the same as before. Define for each $i$; $0 \leq i \leq 5 $,
\begin{align}
\omega_i(x)=\prod\limits_{j\in W_i}(x-\beta^{j}),\nonumber
\end{align}
where $W_i$ denote the Whiteman's cyclotomic classes of order 6.
Among the $nth$ roots of unity $\beta^{i}$, where $0 \leq i \leq n-1$, the $n_2$ elements $\beta^{i}, i \in P \cup \{0\},$ are $n_2$th roots of unity and the $n_1$ elements $\beta^{i}, i \in Q \cup \{0\},$ are $n_1$th roots of unity. Hence,\\
$$x^{n_2}-1=\prod\limits_{i\in P \cup \{0\}}(x-\beta^{i})$$ $\ \mathrm{and}$ $$x^{n_1}-1=\prod\limits_{i\in Q \cup \{0\}}(x-\beta^{i}).$$ \\
Then we have
$ x^n-1 = \prod\limits_{i=0}^{n-1}(x-\beta^{i}) = \frac{(x^{n_1}-1)(x^{n_2}-1)}{x-1}\omega(x),$ where $\omega(x) = \prod\limits_{i=0}^{5} \omega_i(x).$ Also, it can be written as $ x^n-1=\frac{(x^{n_1}-1)(x^{n_2}-1)}{x-1}d_0(x)d_1(x)$, where $d_0(x)=\prod\limits_{i\in D_0}(x-\beta^{i})$ and $d_1(x)=\prod\limits_{i\in D_1}(x-\beta^{i}).$
It is straightforward to prove that if $q \in D_0 $, then $d_i(x) \in \mathrm{GF}(q)$  for all $i .$ 

Now we are ready to compute the generator polynomial and the linear complexity of the sequence $\lambda^\infty$ (defined in  Eq.(\ref{si})). For this,
 let $\bigtriangleup_{1}=\frac{n_{1}+1}{2}\ (\mathrm{mod}\ p)$,\  $\bigtriangleup_{2}=\frac{n_{2}-1}{2}\ (\mathrm{mod}\ p)$ and 
$\bigtriangleup=\frac{(n_{1}+1)(n_{2}-1)}{2}\ (\mathrm{mod}\ p)$. We have the following theorem.
\begin{theorem}\label{gen}
 {\rm{(1)}} When $n\equiv 7\ (\mathrm{mod}\ 12)$ and $\frac{n+1}{4} \notequiv 0 \ (\mathrm{mod}\ p)$ or $n\equiv 1\ (\mathrm{mod}\ 12)$ and $\frac{n-1}{4} \notequiv 0 \ (\mathrm{mod}\ p),$ then the generator polynomial $g_\lambda(x)$ and the linear span $L_\lambda$ of the sequence $\lambda^\infty$ {\rm{(}}defined in  Eq.{\rm{(\ref{si}))}} are given by
 $$
g_\lambda(x)=\left\{
\begin{array}{llll}
x^{n}-1,\ \ \ \ \ \ \ \ \ \ \ \mathrm{if}\ \bigtriangleup_{1}\neq0,\ \bigtriangleup_{2}\neq0,\ \bigtriangleup\neq0\\
\frac{x^{n}-1}{x-1},\ \ \ \ \ \ \ \ \ \ \ \ \ \mathrm{if}\ \bigtriangleup_{1}\neq0,\ \bigtriangleup_{2}\neq0,\ \bigtriangleup=0\\
\frac{x^{n}-1}{x^{n_{2}-1}},\ \ \ \ \ \ \ \ \ \ \ \  \mathrm{if}\ \bigtriangleup_{1}=0,\ \bigtriangleup_{2}\neq0\\
\frac{x^{n}-1}{x^{n_{1}-1}},\ \ \ \ \ \ \ \ \ \ \ \ \mathrm{if}\ \bigtriangleup_{1}\neq0,\ \bigtriangleup_{2}=0\\
\frac{(x^{n}-1)(x-1)}{(x^{n_{1}-1})(x^{n_{2}}-1)},\ \ \mathrm{if}\ \bigtriangleup_{1}=\bigtriangleup_{2}=0.
\end{array}
\right.
$$
and
 $$
L_\lambda(x)=\left\{
\begin{array}{llll}
n,\ \ \ \ \ \ \ \ \ \ \ \ \ \ \ \ \ \ \ \ \ \mathrm{if}\ \bigtriangleup_{1}\neq0,\ \bigtriangleup_{2}\neq0,\ \bigtriangleup\neq0\\
n-1,\ \ \ \ \ \ \ \ \ \ \ \ \ \ \ \  \mathrm{if}\ \bigtriangleup_{1}\neq0,\ \bigtriangleup_{2}\neq0,\ \bigtriangleup=0\\
n-n_{2},\ \ \ \ \ \ \ \ \ \ \ \ \ \ \ \mathrm{if}\ \bigtriangleup_{1}=0,\ \bigtriangleup_{2}\neq0\\
n-n_{1},\ \ \ \ \ \ \ \ \ \ \ \ \ \ \ \mathrm{if}\ \bigtriangleup_{1}\neq0,\ \bigtriangleup_{2}=0\\
n-(n_1+n_2-1),\ \mathrm{if}\ \bigtriangleup_{1}=\bigtriangleup_{2}=0.
\end{array}
\right.
$$
In this case, the cyclic code $C_\lambda$ over $\mathrm{GF}(q)$ defined by the two-prime WGCS-I of order 6 has generator polynomial $g_\lambda(x)$ as above and parameters $[n,k,d]$, where the dimension $k=n-\mathrm{deg}(g_\lambda(x))$. 
{\rm{(2)}} When $n\equiv 7\ (\mathrm{mod}\ 12)$ and $\frac{n+1}{4} \equiv 0 \ (\mathrm{mod}\ p)$ or $n\equiv 1\ (\mathrm{mod}\ 12)$ and $\frac{n-1}{4} \equiv 0 \ (\mathrm{mod}\ p),$ then the generator polynomial $g_\lambda(x)$ and the linear span $L_\lambda$ of the sequence $\lambda^\infty$ are given by
$$
g_\lambda(x)=\left\{
\begin{array}{llll}
\frac{x^{n}-1}{d_0(x)},\ \ \ \ \ \ \ \ \ \ \ \ \ \ \ \ \ \ \mathrm{if}\ \bigtriangleup_{1}\neq0,\ \bigtriangleup_{2}\neq0,\ \bigtriangleup\neq0,\ \Lambda(\beta)=0\\
\frac{x^{n}-1}{d_1(x)},\ \ \ \ \ \ \ \ \ \ \ \ \ \ \ \ \ \ \mathrm{if}\ \bigtriangleup_{1}\neq0,\ \bigtriangleup_{2}\neq0,\ \bigtriangleup\neq0,\ \Lambda(\beta)=-1\\
\frac{x^{n}-1}{(x-1)d_0(x)},\ \ \ \ \ \ \ \ \ \ \ \ \mathrm{if}\ \bigtriangleup_{1}\neq0,\ \bigtriangleup_{2}\neq0,\ \bigtriangleup=0,\ \Lambda(\beta)=0\\
\frac{x^{n}-1}{(x-1)d_1(x)},\ \ \ \ \ \ \ \ \ \ \ \ \mathrm{if}\ \bigtriangleup_{1}\neq0,\ \bigtriangleup_{2}\neq0,\ \bigtriangleup=0,\ \Lambda(\beta)=-1\\
\frac{x^{n}-1}{x^{n_{2}-1}d_0(x)},\ \ \ \ \ \ \ \ \ \ \ \ \mathrm{if}\ \bigtriangleup_{1}=0,\ \bigtriangleup_{2}\neq0,\ \Lambda(\beta)=0\\
\frac{x^{n}-1}{x^{n_{2}-1}d_1(x)},\ \ \ \ \ \ \ \ \ \ \ \ \mathrm{if}\ \bigtriangleup_{1}=0,\ \bigtriangleup_{2}\neq0,\ \Lambda(\beta)=-1\\
\frac{x^{n}-1}{x^{n_{1}-1}d_0(x)},\ \ \ \ \ \ \ \ \ \ \ \ \mathrm{if}\ \bigtriangleup_{1}\neq0,\ \bigtriangleup_{2}=0,\ \Lambda(\beta)=0\\
\frac{x^{n}-1}{x^{n_{1}-1}d_1(x)},\ \ \ \ \ \ \ \ \ \ \ \ \mathrm{if}\ \bigtriangleup_{1}\neq0,\ \bigtriangleup_{2}=0,\ \Lambda(\beta)=-1\\
\frac{(x^{n}-1)(x-1)}{(x^{n_{1}-1})(x^{n_{2}}-1)d_0(x)},\ \mathrm{if}\ \bigtriangleup_{1}=\bigtriangleup_{2}=0,\ \Lambda(\beta)=0\\
\frac{(x^{n}-1)(x-1)}{(x^{n_{1}-1})(x^{n_{2}}-1)d_1(x)},\ \mathrm{if}\ \bigtriangleup_{1}=\bigtriangleup_{2}=0,\ \Lambda(\beta)=-1.
\end{array}
\right.
$$
and
$$
L_\lambda(x)=\left\{
\begin{array}{llll}
n-\frac{(n_1-1)(n_2-1)}{2},\ \ \ \ \ \mathrm{if}\ \bigtriangleup_{1}\neq0,\ \bigtriangleup_{2}\neq0,\ \bigtriangleup\neq0, \Lambda(\beta)=0\ \mathrm{or}\ \Lambda(\beta)=-1\\

n-\frac{({n_1}-1)({n_2}-1)+2}{2},\ \   \mathrm{if}\ \bigtriangleup_{1}\neq0,\ \bigtriangleup_{2}\neq0,\ \bigtriangleup=0,\Lambda(\beta)=0 \ \mathrm{or}\ \Lambda(\beta)=-1\\

n-\frac{({n_1}+1)({n_2}-1)+2}{2},\ \  \mathrm{if}\ \bigtriangleup_{1}=0,\ \bigtriangleup_{2}\neq0,\Lambda(\beta)=0\ \mathrm{or}\ \Lambda(\beta)=-1\\

n-\frac{({n_1}-1)({n_2}+1)+2}{2},\ \   \mathrm{if}\ \bigtriangleup_{1}\neq0,\ \bigtriangleup_{2}=0,\Lambda(\beta)=0\ \mathrm{or}\ \Lambda(\beta)=-1\\

n-\frac{({n_1}+1)({n_2}+1)-2}{2},\ \ \mathrm{if}\ \bigtriangleup_{1}=\bigtriangleup_{2}=0,\Lambda(\beta)=0\ \mathrm{or}\ \Lambda(\beta)=-1.\\

\end{array}
\right.
$$
In this case, the cyclic code $C_\lambda$ over $GF(q)$ defined by the two-prime WGCS-I of order 6  has generator polynomial $g_\lambda(x)$ as above and parameters $[n,k,d]$, where the dimension $k=n-\mathrm{deg}(g_\lambda(x))$.
\end{theorem}

\begin{proof} 
 \rm{(1)} When $n\equiv 7\ (\mathrm{mod}\ 12)$ and $\frac{n+1}{4} \notequiv 0 \ (\mathrm{mod}\ p)$ or $n\equiv 1\ (\mathrm{mod}\ 12)$ and $\frac{n-1}{4} \notequiv 0 \ (\mathrm{mod}\ p),$ then by Lemma \ref{imp}, we have $\Lambda(\beta)\neq 0,-1.$ Therefore, from Lemma \ref{value}, $\Lambda(\beta^t)=0$ only when $t$ is in $P$ or $Q$ or both. So, the conclusion on the generator polynomial $g_\lambda(x)$ of cyclic code $C_\lambda$ over $\mathrm{GF}(q)$ defined by the sequence $\lambda^\infty$ follows from Eq.(\ref{lam1}) and Lemma \ref{value}. The linear span of the sequence $\lambda^{\infty}$ is equal to $\mathrm{deg}(g_\lambda(x))$.\\
\rm{(2)} When $n\equiv 7\ (\mathrm{mod}\ 12)$ and $\frac{n+1}{4} \equiv 0 \ (\mathrm{mod}\ p)$ or $n\equiv 1\ (\mathrm{mod}\ 12)$ and $\frac{n-1}{4} \equiv 0 \ (\mathrm{mod}\ p),$ then by Lemma \ref{imp}, we have $\Lambda(\beta)\in \{0,-1\}$ and $d_i(x)\in \mathrm{GF}(q)[x]$ for each $i$ if $q \in D_0$. So, the conclusion on the generator polynomial $g_\lambda(x)$ of cyclic code $C_\lambda$ over $\mathrm{GF}(q)$ defined by the sequence $\lambda^\infty$ follows from Eq.(\ref{lam1}), Lemma \ref{D0},\ref{lemma 3.7} and \ref{value}. The linear span of the sequence $\lambda^{\infty}$ is equal to $\mathrm{deg}(g_\lambda(x))$.
\end{proof}

The following corollaries follows from Theorem \ref{gen}, Lemma \ref{imp} and \ref{D0} and give the conclusions on the linear span and generator polynomial of the sequence $\lambda^\infty$ (defined  in Eq.(\ref{si})).

\begin{corollary}
 Let $q=2.$ We have the following conclusions:\\
 {\rm{(1)}} If $n_1\equiv 13\ (\mathrm{mod}\ 24)$ and $n_2\equiv 7\ (\mathrm{mod}\ 24)$ or $n_1\equiv 1\ (\mathrm{mod}\ 24)$ and $n_2\equiv 19\ (\mathrm{mod}\ 24),$ we have
 \begin{align}
  g_\lambda(x) = \frac{x^n-1}{x-1}\ \ \mathrm{and} \ \ \ L_\lambda=n-1.\nonumber
 \end{align}
In this case, the cyclic code $C_\lambda$  over $\mathrm{GF}(q)$ defined by the sequence $\lambda^{\infty}$ has parameters $[n,1,n-1]$ and generator polynomial $g_\lambda(x)$ as above.\\
 {\rm{(2)}} If $n_1\equiv 7\ (\mathrm{mod}\ 24)$ and $n_2\equiv 19\ (\mathrm{mod}\ 24)$ or $n_1\equiv 19\ (\mathrm{mod}\ 24)$ and $n_2\equiv 7\ (\mathrm{mod}\ 24),$ we have
 \begin{align}
  g_\lambda(x) = \frac{x^n-1}{x^{n_2}-1}\ \ \mathrm{and} \ \ \ L_\lambda=n-n_2.\nonumber
 \end{align}
In this case, the cyclic code $C_\lambda$  over $\mathrm{GF}(q)$ defined by the sequence $\lambda^{\infty}$ has parameters $[n,n_2,n_1]$  $($From Theorem $\ref{notinw0}, d=n_1)$  and generator polynomial $g_\lambda(x)$ as above.\\ 
 {\rm{(3)}} If $n_1\equiv 7\ (\mathrm{mod}\ 24)$ and $n_2\equiv 13\ (\mathrm{mod}\ 24)$ or $n_1\equiv 19\ (\mathrm{mod}\ 24)$ and $n_2\equiv 1\ (\mathrm{mod}\ 24),$ we have
 \begin{align}
  g_\lambda(x) = \frac{(x^n-1)(x-1)}{(x^{n_1}-1(x^{n_2}-1)}\ \ \ \mathrm{and} \ \ L_\lambda=n-(n_1+n_2-1).\nonumber
 \end{align}
In this case, the cyclic code $C_\lambda$  over $\mathrm{GF}(q)$ defined by the sequence $\lambda^{\infty}$ has parameters $[n,n_1+n_2-1,d]$ and generator polynomial $g_\lambda(x)$ as above.\\ 
{\rm{(4)}} If $n_1\equiv 1\ (\mathrm{mod}\ 24)$ and $n_2\equiv 7\ (\mathrm{mod}\ 24)$ or $n_1\equiv 13\ (\mathrm{mod}\ 24)$ and $n_2\equiv 19\ (\mathrm{mod}\ 24),$ we have
 \begin{align}
  g_\lambda(x) =\left\{
  \begin{array}{ll}
  \frac{(x^n-1)}{(x-1)d_0(x)},\ \ \ if\ \ \Lambda(\beta)=0\\
  \frac{(x^n-1)}{(x-1)d_1(x)},\ \ \ if\ \ \Lambda(\beta)=1\nonumber
  \end{array}
  \right.
  \ \ \mathrm{and}\ \ \ L_\lambda=n-\frac{(n_1-1)(n_2-1)+2}{2}.
 \end{align}
In this case, the cyclic code $C_\lambda$  over $\mathrm{GF}(q)$ defined by the sequence $\lambda^{\infty}$ has parameters $[n,\frac{(n_1-1)(n_2-1)+2}{2},d]$ and generator polynomial $g_\lambda(x)$ as above.\\
{\rm{(5)}} If $n_1\equiv 7\ (\mathrm{mod}\ 24)$ and $n_2\equiv 7\ (\mathrm{mod}\ 24)$ or $n_1\equiv 19\ (\mathrm{mod}\ 24)$ and $n_2\equiv 19\ (\mathrm{mod}\ 24),$ we have
 \begin{align}
  g_\lambda(x) = \left\{
  \begin{array}{ll}
  \frac{(x^n-1)}{(x^{n_2}-1)d_0(x)},\ \ \ if\ \ \Lambda(\beta)=0\\
  \frac{(x^n-1)}{(x^{n_2}-1)d_1(x)},\ \ \ if\ \ \Lambda(\beta)=1\nonumber
  \end{array}
  \right.
  \ \ \mathrm{and}\ \ \ L_\lambda=n-\frac{(n_1+1)(n_2-1)+2}{2}.
 \end{align}
In this case, the cyclic code $C_\lambda$  over $\mathrm{GF}(q)$ defined by the sequence $\lambda^{\infty}$ has parameters $[n,\frac{(n_1+1)(n_2-1)+2}{2},d]$ and generator polynomial $g_\lambda(x)$ as above. \\
{\rm{(6)}} If $n_1\equiv 7\ (\mathrm{mod}\ 24)$ and $n_2\equiv 1\ (\mathrm{mod}\ 24)$ or $n_1\equiv 19\ (\mathrm{mod}\ 24)$ and $n_2\equiv 13\ (\mathrm{mod}\ 24),$ we have
 \begin{align}
g_\lambda(x) =\left\{
\begin{array}{ll}
\frac{(x^n-1)(x-1)}{(x^{n_1}-1)(x^{n_2}-1)d_0(x)},\ \ \ if\ \ \Lambda(\beta)=0\\
\frac{(x^n-1)(x-1)}{(x^{n_1}-1)(x^{n_2}-1)d_1(x)},\ \ \ if\ \ \Lambda(\beta)=1\nonumber 
\end{array}
\right.
\ \ \mathrm{and}\ \ \ L_\lambda=n-\frac{(n_1+1)(n_2+1)-2}{2}.
 \end{align}
In this case, the cyclic code $C_\lambda$  over $\mathrm{GF}(q)$ defined by the sequence $\lambda^{\infty}$ has parameters $[n,\frac{(n_1+1)(n_2+1)-2}{2},d]$ and generator polynomial $g_\lambda(x)$ as above. 
\end{corollary}

If $q=3$, then we have only one possibility: $n_1\equiv 7\ (\mathrm{mod}\ 12)$ and $n_2\equiv 7\ (\mathrm{mod}\ 12)$. 
\begin{corollary}
Let $q=3$ and  $n_1\equiv 7\ (\mathrm{mod}\ 12)$ and $n_2\equiv 7\ (\mathrm{mod}\ 12)$, we have
 \begin{align}
g_\lambda(x) =\left\{
\begin{array}{ll}
\frac{(x^n-1)}{(x^{n_1}-1)d_0(x)},\ \ \ if\ \ \Lambda(\beta)=0\\
\frac{(x^n-1)}{(x^{n_1}-1)d_1(x)},\ \ \ if\ \ \Lambda(\beta)=1\nonumber 
\end{array}
\right.
\ \ \mathrm{and}\ \ \ L_\lambda=n-\frac{(n_1-1)(n_2+1)+2}{2}.
 \end{align}
In this case, the cyclic code $C_\lambda$  over $\mathrm{GF}(q)$ defined by the sequence $\lambda^{\infty}$ has parameters $[n,\frac{(n_1-1)(n_2+1)+2}{2},d]$ and generator polynomial $g_\lambda(x)$ as above. 
\end{corollary}

\begin{corollary}
 Let $q=5.$ We have the following conclusions:\\
 {\rm{(1)}} If $n_1\equiv 1\ (\mathrm{mod}\ 60)$ and $n_2\equiv 43\ (\mathrm{mod}\ 60)$ or $n_1\equiv 1\ (\mathrm{mod}\ 60)$ and $n_2\equiv 7\ (\mathrm{mod}\ 60)$ or $n_1\equiv 31\ (\mathrm{mod}\ 60)$ and $n_2\equiv 43\ (\mathrm{mod}\ 60)$ or $n_1\equiv 31\ (\mathrm{mod}\ 60)$ and $n_2\equiv 7\ (\mathrm{mod}\ 60)$ or $n_1\equiv 31\ (\mathrm{mod}\ 60)$ and $n_2\equiv 13\ (\mathrm{mod}\ 60)$ or $n_1\equiv 31\ (\mathrm{mod}\ 60)$ and $n_2\equiv 37\ (\mathrm{mod}\ 60),$ we have
 \begin{align}
  g_\lambda(x) = x^n-1\ \ \mathrm{and} \ \ \ L_\lambda=n.\nonumber
 \end{align}
In this case, the cyclic code $C_\lambda$  over $\mathrm{GF}(q)$ defined by the sequence $\lambda^{\infty}$ has parameters $[n,0,0]$ and generator polynomial $g_\lambda(x)$ as above.\\
 {\rm{(2)}} If $n_1\equiv 19\ (\mathrm{mod}\ 60)$ and $n_2\equiv 13\ (\mathrm{mod}\ 60)$ or $n_1\equiv 19\ (\mathrm{mod}\ 60)$ and $n_2\equiv 7\ (\mathrm{mod}\ 60)$ or $n_1\equiv 19\ (\mathrm{mod}\ 60)$ and $n_2\equiv 43\ (\mathrm{mod}\ 60)$ or $n_1\equiv 49\ (\mathrm{mod}\ 60)$ and $n_2\equiv 43\ (\mathrm{mod}\ 60)$ or $n_1\equiv 49\ (\mathrm{mod}\ 60)$ and $n_2\equiv 7\ (\mathrm{mod}\ 60)$, we have
 \begin{align}
  g_\lambda(x) = \frac{x^n-1}{x^{n_2}-1}\ \ \mathrm{and} \ \ \ L_\lambda=n-n_2.\nonumber
 \end{align}
In this case, the cyclic code $C_\lambda$  over $\mathrm{GF}(q)$ defined by the sequence $\lambda^{\infty}$ has parameters $[n,n_2,n_1]$ $($From Theorem $\ref{notinw0}, d=n_1)$  and generator polynomial $g_\lambda(x)$ as above.\\ 
 {\rm{(3)}} If $n_1\equiv 43\ (\mathrm{mod}\ 60)$ and $n_2\equiv 1\ (\mathrm{mod}\ 60)$ or $n_1\equiv 7\ (\mathrm{mod}\ 60)$ and $n_2\equiv 1\ (\mathrm{mod}\ 60)$ or $n_1\equiv 43\ (\mathrm{mod}\ 60)$ and $n_2\equiv 31\ (\mathrm{mod}\ 60)$ or $n_1\equiv 7\ (\mathrm{mod}\ 60)$ and $n_2\equiv 31\ (\mathrm{mod}\ 60)$ or $n_1\equiv 37\ (\mathrm{mod}\ 60)$ and $n_2\equiv 31\ (\mathrm{mod}\ 60)$ or $n_1\equiv 13\ (\mathrm{mod}\ 60)$ and $n_2\equiv 31\ (\mathrm{mod}\ 60),$ we have
 \begin{align}
  g_\lambda(x) = \frac{x^n-1}{x^{n_1}-1}\ \ \ \mathrm{and} \ \ L_\lambda=n-n_1.\nonumber
 \end{align}
In this case, the cyclic code $C_\lambda$  over $\mathrm{GF}(q)$ defined by the sequence $\lambda^{\infty}$ has parameters $[n,n_1,n_2]$ $($From Theorem $\ref{notinw0}, d=n_2)$  and generator polynomial $g_\lambda(x)$ as above.\\ 
{\rm{(4)}} If $n_1\equiv 1\ (\mathrm{mod}\ 60)$ and $n_2\equiv 19\ (\mathrm{mod}\ 60)$ or $n_1\equiv 31\ (\mathrm{mod}\ 60)$ and $n_2\equiv 49\ (\mathrm{mod}\ 60)$ or $n_1\equiv 13\ (\mathrm{mod}\ 60)$ and $n_2\equiv 43\ (\mathrm{mod}\ 60)$ or $n_1\equiv 37\ (\mathrm{mod}\ 60)$ and $n_2\equiv 7\ (\mathrm{mod}\ 60)$ or $n_1\equiv 43\ (\mathrm{mod}\ 60)$ and $n_2\equiv 13\ (\mathrm{mod}\ 60)$ or $n_1\equiv 7\ (\mathrm{mod}\ 60)$ and $n_2\equiv 37\ (\mathrm{mod}\ 60)$ or $n_1\equiv 31\ (\mathrm{mod}\ 60)$ and $n_2\equiv 19\ (\mathrm{mod}\ 60)$ or $n_1\equiv 13\ (\mathrm{mod}\ 60)$ and $n_2\equiv 7\ (\mathrm{mod}\ 60)$ or $n_1\equiv 37\ (\mathrm{mod}\ 60)$ and $n_2\equiv 43\ (\mathrm{mod}\ 60)$  we have
 \begin{align}
  g_\lambda(x) =\left\{
  \begin{array}{ll}
  \frac{(x^n-1)}{d_0(x)},\ \ \ if\ \ \Lambda(\beta)=0\\
  \frac{(x^n-1)}{d_1(x)},\ \ \ if\ \ \Lambda(\beta)=1\nonumber
  \end{array}
  \right.
  \ \ \mathrm{and}\ \ \ L_\lambda=n-\frac{(n_1-1)(n_2-1)}{2}.
 \end{align}
In this case, the cyclic code $C_\lambda$  over $\mathrm{GF}(q)$ defined by the sequence $\lambda^{\infty}$ has parameters $[n,\frac{(n_1-1)(n_2-1)}{2},d]$ and generator polynomial $g_\lambda(x)$ as above.\\
{\rm{(5)}} If $n_1\equiv 19\ (\mathrm{mod}\ 60)$ and $n_2\equiv 19\ (\mathrm{mod}\ 60)$ or $n_1\equiv 19\ (\mathrm{mod}\ 60)$ and $n_2\equiv 49\ (\mathrm{mod}\ 60)$ or $n_1\equiv 49\ (\mathrm{mod}\ 60)$ and $n_2\equiv 19\ (\mathrm{mod}\ 60)$  we have
 \begin{align}
  g_\lambda(x) = \left\{
  \begin{array}{ll}
  \frac{(x^n-1)}{(x^{n_2}-1)d_0(x)},\ \ \ if\ \ \Lambda(\beta)=0\\
  \frac{(x^n-1)}{(x^{n_2}-1)d_1(x)},\ \ \ if\ \ \Lambda(\beta)=1\nonumber
  \end{array}
  \right.
  \ \ \mathrm{and}\ \ \ L_\lambda=n-\frac{(n_1+1)(n_2-1)+2}{2}.
 \end{align}
In this case, the cyclic code $C_\lambda$  over $\mathrm{GF}(q)$ defined by the sequence $\lambda^{\infty}$ has parameters $[n,\frac{(n_1+1)(n_2-1)+2}{2},d]$ and generator polynomial $g_\lambda(x)$ as above. \\
{\rm{(6)}} If $n_1\equiv 1\ (\mathrm{mod}\ 60)$ and $n_2\equiv 31\ (\mathrm{mod}\ 60)$ or $n_1\equiv 31\ (\mathrm{mod}\ 60)$ and $n_2\equiv 1\ (\mathrm{mod}\ 60)$, we have
 \begin{align}
  g_\lambda(x) = \left\{
  \begin{array}{ll}
  \frac{(x^n-1)}{(x^{n_1}-1)d_0(x)},\ \ \ if\ \ \Lambda(\beta)=0\\
  \frac{(x^n-1)}{(x^{n_1}-1)d_1(x)},\ \ \ if\ \ \Lambda(\beta)=1\nonumber
  \end{array}
  \right.
  \ \ \mathrm{and}\ \ \ L_\lambda=n-\frac{(n_1-1)(n_2+1)+2}{2}.
 \end{align}
In this case, the cyclic code $C_\lambda$  over $\mathrm{GF}(q)$ defined by the sequence $\lambda^{\infty}$ has parameters $[n,\frac{(n_1-1)(n_2+1)+2}{2},d]$ and generator polynomial $g_\lambda(x)$ as above. \\
{\rm{(7)}} If $n_1\equiv 19\ (\mathrm{mod}\ 60)$ and $n_2\equiv 31\ (\mathrm{mod}\ 60)$ or $n_1\equiv 19\ (\mathrm{mod}\ 60)$ and $n_2\equiv 1\ (\mathrm{mod}\ 60)$ or $n_1\equiv 49\ (\mathrm{mod}\ 60)$ and $n_2\equiv 31\ (\mathrm{mod}\ 60)$ we have
 \begin{align}
g_\lambda(x) =\left\{
\begin{array}{ll}
\frac{(x^n-1)(x-1)}{(x^{n_1}-1)(x^{n_2}-1)d_0(x)},\ \ \ if\ \ \Lambda(\beta)=0\\
\frac{(x^n-1)(x-1)}{(x^{n_1}-1)(x^{n_2}-1)d_1(x)},\ \ \ if\ \ \Lambda(\beta)=1\nonumber 
\end{array}
\right.
\ \ \mathrm{and}\ \ \ L_\lambda=n-\frac{(n_1+1)(n_2+1)-2}{2}.
 \end{align}
In this case, the cyclic code $C_\lambda$  over $\mathrm{GF}(q)$ defined by the sequence $\lambda^{\infty}$ has parameters $[n,\frac{(n_1+1)(n_2+1)-2}{2},d]$ and generator polynomial $g_\lambda(x)$ as above. 
\end{corollary}

\section{The minimum distance of the cyclic codes}\label{section 5}
In this section, we determine the lower bounds on the minimum distance of some of the cyclic codes of this paper.

\begin{theorem}\label{notinw0}
\cite{ding12} Let $C_{i}$ denote the cyclic code over $\mathrm{GF}(q)$ with the generator polynomial $g_{i}(x)=\frac{x^{n}-1}{x^{n_{i}}-1}.$ 
The cyclic code $C_{i}$  has parameters $[n,n_{i},d_{i}]$, where $d_{i}=n_{i-(-1)^{i}}$ and $i=1,2$.
\end{theorem}

\begin{theorem}\label{notinw0.1}
\cite{ding12} Let $C_{(n_{1},n_{2},q)}$ denote the cyclic code over $\mathrm{GF}(q)$ with the generator polynomial $g(x)=\frac{(x^{n}-1)(x-1)}{(x^{n_{1}}-1)(x^{n_{2}}-1)}$. 
 The cyclic code $C_{(n_{1},n_{2},q)}$ has parameters $[n,n_{1}+n_{2}-1,d_{(n_{1},n_{2},q)}]$, where $d_{(n_{1},n_{2},q)}=\mathrm{min}(n_{1},n_{2})$.
\end{theorem}

\begin{theorem}\label{minw0.1}
Assume that $q\in D_{0}$. Let $C^{(i,j)}$ denote the cyclic code over $\mathrm{GF}(q)$ with the generator polynomial 
$g^{(i,j)}(x)=\frac{x^{n}-1}{(x^{n_{i}}-1)d_{j}(x)}$ and let $d^{(i,j)}$ denote the minimum distance of this code, where $i\in\{1,2\}$ and $j\in\{0,1\}$.
The cyclic code $C^{(i,j)}$ has parameters $[n,n_{i}+\frac{(n_{1}-1)(n_{2}-1)}{2},d^{(i,j)}]$, 
where $d^{(i,j)}\geq\lceil\sqrt{n_{i-(-1)^{i}}}\rceil$.\\
If $-1\in D_1,$ we have ${(d^{(i,j)}})^{2}- d^{(i,j)} + 1\geq n_{i-(-1)^{i}}.$
\end{theorem}

\begin{proof} Let $c(x)\in \mathrm{GF}(q)[x]/(x^{n}-1)$ be a codeword of Hamming weight $\omega$ in $C^{(i,j)}$. 
Take any $r\in D_{1}$. The cyclic code $c(x^{r})$ is a codeword of Hamming weight $\omega$ in $C^{(i,(j+1)\ \mathrm{mod}\ 2)}$. 
It then follows that $d^{(i,j)}=d^{(i,(j+1)\ \mathrm{mod}\ 2)}.$ 
Let $c(x)\in \mathrm{GF}(q)[x]/(x^{n}-1)$ be a codeword of minimum weight in $C^{(i,j)}$. 
Then $c(x^{r})$ is a codeword of same weight in $C^{(i,(j+1)\ \mathrm{mod}\ 2)}$. 
Hence, $c(x)c(x^{r})$ is a codeword of $C_{i}$, where $C_{i}$ denote the cyclic code over $\mathrm{GF}(q)$ with the generator polynomial $g_i(x)=\frac{x^n-1}{x^{n_i}-1}$ and 
minimum distance $d_{i}=n_{i-(-1)^{i}}$. Hence, from Theorem \ref{notinw0}, we have ${(d^{(i,j)}})^{2}\geq d_{i}=n_{i-(-1)^{i}},$ and
${(d^{(i,j)}})^{2}- d^{(i,j)} + 1\geq n_{i-(-1)^{i}}$ if $-1\in D_1.$
\end{proof}
\begin{theorem}\label{minw0.2}
Assume that $q\in D_{0}$. Let $C_{(n_1,n_2)}^{(j)}$ denote the cyclic code over $\mathrm{GF}(q)$ 
with the generator polynomial $g_{(n_1,n_2)}^{(j)}(x)=\frac{(x^{n}-1)(x-1)}{(x^{n_{1}}-1)(x^{n_{2}}-1)d_{j}(x)}$ 
and let $d_{n_1,n_2}^{(j)}$ denote the minimum distance of this code, where $i\in\{1,2\}$ and $j\in\{0,1\}$.
The cyclic code $C_{(n_1,n_2)}^{(j)}$ has parameters $[n,n_{1}+n_{2}-1+\frac{(n_{1}-1)(n_{2}-1)}{2},d_{(n_1,n_2)}^{(j)}]$, 
where $d_{(n_1,n_2)}^{(j)}\geq\lceil\sqrt{\mathrm{min}(n_{1},n_{2})}\rceil$.\\
If $-1\in D_1,$ we have ${(d_{(n_1,n_2)}^{(j)}})^{2}- d_{(n_1,n_2)}^{(j)} + 1\geq \mathrm{min}(n_{1},n_2).$
\end{theorem}

\begin{proof} Let $c(x)\in \mathrm{GF}(q)[x]/(x^{n}-1)$ be a codeword of Hamming weight $\omega$ in $C_{(n_1,n_2)}^{(j)}$. 
Take any $r\in D_{1}$. The cyclic code $c(x^{r})$ is a codeword of Hamming weight $\omega$ in $C_{(n_1,n_2)}^{((j+1)\ \mathrm{mod}\ 2)}$. 
It then follows that $d_{(n_1,n_2)}^{(j)}=d_{(n_1,n_2)}^{((j+1)\ \mathrm{mod}\ 2)}.$ 
Let $c(x)\in \mathrm{GF}(q)[x]/(x^{n}-1)$ be a codeword of minimum weight in  $C_{(n_1,n_2)}^{(j)}$. 
Then $c(x^{r})$ is a codeword of same weight in $C_{(n_1,n_2)}^{((j+1)\ \mathrm{mod}\ 2)}$. 
Hence, $c(x)c(x^{r})$ is a codeword of $C_{(n_1,n_2,q)}$, where $C_{(n_1,n_2,q)}$ denote the cyclic code
over $\mathrm{GF}(q)$ with the generator polynomial $g(x)=\frac{(x^n-1)(x-1)}{(x^{n_1}-1)(x^{n_2}-1)}$ and 
minimum distance $d_{(n_1,n_2,q)}=\mathrm{min} (n_1, n_2).$ 
Hence, from Theorem \ref{notinw0.1}, we have ${(d_{(n_1,n_2)}^{(j)}})^{2}\geq d_{(n_1,n_2)^j}=\mathrm{min}(n_{1},n_2),$ and
${(d_{(n_1,n_2)}^{(j)}})^{2}- d_{(n_1,n_2)}^{(j)} + 1\geq \mathrm{min}(n_{1},n_2)$ if $-1\in D_1.$
\end{proof}
\begin{example}\label{ex1}
 Let $(p,m,n_1,n_2)=(2,1,7,31).$ Then $q=2,\ n=217$ and $C_\lambda$ is a $[217,121]$ cyclic code over $\mathrm{GF}(q)$ with generator polynomial
$g(x)=\frac{x^{217}-1}{(x^{31}-1)d_1(x)}=x^{96}+x^{94}+x^{91}+x^{87}+x^{86}+x^{85}
+x^{83}+x^{81}+x^{80}+x^{78}+x^{77}+x^{75}+x^{72}+x^{69}+x^{67}+x^{65}+x^{64}+x^{63}+x^{60}+x^{58}+x^{55}+x^{53}+x^{52}+x^{51}+x^{48}+x^{45}+x^{44}+x^{43}+x^{41}+x^{38}+x^{36}+x^{33}+x^{32}+x^{31}+x^{29}+x^{27}+x^{24}+x^{21}+x^{19}+x^{18}+x^{16}+x^{15}+x^{13}+x^{11}+x^{10}+x^9+x^5+x^2+1.$
We did some computation and our computation shows that upper bound on the minimum distance for this binary code is $31$. From Theorem \ref{minw0.1}, we have the lower bound on the minimum distance for this binary code is $3$.
\end{example}

\begin{example}\label{ex2}
 Let $(p,m,n_1,n_2)=(2,1,7,31).$ Then $q=3,\ n=217$ and $C_\lambda$ is a $[217,97]$ cyclic code over $\mathrm{GF}(q)$ with generator polynomial
$g(x)=\frac{x^{217}-1}{(x^{7}-1)d_1(x)}=x^{120}+2x^{115}+x^{113}+2x^{109}+2x^{108}+x^{106}+x^{105}+x^{104}+2x^{102}+2x^{100}+x^{98}+2x^{96}+x^{95}+x^{92}+
x^{90}+x^{88}+2x^{87}+2x^{85}+x^{83}+2x^{81}+x^{79}+x^{78}+2x^{77}+x^{76}+2x^{75}+2x^{74}+2x^{71}+2x^{70}+x^{69}+x^{67}+2x^{66}+2x^{65}+x^{64}+2x^{61}+x^{60}+2x^{59}+x^{56}+2x^{55}+2x^{54}+x^{53}+x^{51}+2x^{50}+2x^{49}+2x^{46}+2x^{45}+x^{44}+2x^{43}+x^{42}+x^{41}+2x^{39}+x^{37}+2x^{35}+2x^{33}+x^{32}+x^{30}+x^{28}+x^{25}+2x^{24}+x^{22}+2x^{20}+2x^{18}+x^{16}+x^{15}+x^{14}+2x^{12}+2x^{11}+x^{7}+2x^{5}+1.$
We did some computation and our computation shows that upper bound on the minimum distance for this ternary code is $58$. From Theorem \ref{minw0.1}, we have the lower bound on the minimum distance for this binary code is $6$.
\end{example}

\begin{example}\label{ex3}
 Let $(p,m,n_1,n_2)=(2,1,7,19).$ Then $q=2,\ n=133$ and $C_\lambda$ is a $[133,19,7]$ cyclic code over $\mathrm{GF}(q)$ with generator polynomial
$g(x)=\frac{(x^{133}-1)}{(x^{19}-1)(x^{13}-1)}=x^{114}+x^{95}+x^{76}+x^{57}+x^{38}+x^{19}+1.$
 This is a bad cyclic code due to its poor minimum distance. The code in this case is bad because $q\notin D_0.$
 \end{example}
 
\begin{conclusion*}
 In this paper, we have computed the autocorrelation values and linear complexities of the two-prime WGCS-I of order 6. We have also constructed the cyclic codes over $\mathrm{GF}(q)$ of WGCS-I of order 6. The autocorrelation value is four-valued if $n\equiv 7\ (\mathrm{mod}\ 12)$ and is five-valued if  $n\equiv 1\ (\mathrm{mod}\ 12)$ for any $|n_2-n_1|$ is multiple of $6$. This two-prime WGCS-I of order 6 has low autocorrelation values.
 In the case $\Lambda(\beta)\notin \{0,1\}$, the least value of linear complexity is $n-(n_1+n_2-1)$ and in the case $\Lambda(\beta)\in \{0,1\}$, the least value of linear complexity is $n-\frac{(n_1-1)(n_2-1)}{2}>\frac{n}{2}$. Our results show that these sequence possesses large linear complexity.
 
 The cyclic codes employed in this paper depend on $n_1, n_2 $ and $q$. When $q\in D_0$, we get a good code. We expect that the codes in Examples {\rm{\ref{ex1}}} and {\rm{\ref{ex2}}} give good codes.
  When $q\notin D_0$, we get a bad code, for example, we get a bad code in Example {\rm{\ref{ex3}}}.
  Finally, we expect that cyclic codes described in this paper can be employed to construct the good cyclic codes of large length.

\end{conclusion*}

\bibliographystyle{plain}
   \bibliography{cyl2ref}

 \end{document}